\documentclass[11pt]{article}
\usepackage[margin=1in]{geometry}
\usepackage{amsmath,amsfonts,amsthm,amssymb}
\usepackage{url}
\usepackage{color}
\usepackage[usenames,dvipsnames,svgnames,table]{xcolor}
\usepackage[colorlinks=true, linkcolor=red, urlcolor=blue, citecolor=gray]{hyperref}
\usepackage{algorithm}
\usepackage{algpseudocode}
\usepackage{float}
\usepackage{threeparttable}
\usepackage{tablefootnote}
\makeatletter

\DeclareMathOperator*{\E}{\mathbb{E}}
\newcommand{\ha}{A^{1/2}}
\DeclareMathOperator{\tr}{tr}
\DeclareMathOperator{\rank}{rank}
\newcommand{\R}{\mathbb{R}}
\DeclareMathOperator*{\argmin}{arg\,min}
\newcommand{\poly}{\mathop\mathrm{poly}}
\DeclareMathOperator{\nnz}{nnz}
\newcommand{\eqdef}{\mathbin{\stackrel{\rm def}{=}}}
\newcommand{\norm}[1]{\|#1\|} 

\makeatletter

\newtheorem*{rep@theorem}{\rep@title}
\newcommand{\newreptheorem}[2]{%
\newenvironment{rep#1}[1]{%
 \def\rep@title{#2 \ref{##1}}%
 \begin{rep@theorem}}%
 {\end{rep@theorem}}}
\makeatother
\newtheorem{theorem}{Theorem}
\newreptheorem{theorem}{Theorem}

\newtheorem{corollary}[theorem]{Corollary}
\newtheorem{lemma}[theorem]{Lemma}
\newtheorem*{lemma*}{Lemma}
\newreptheorem{lemma}{Lemma}

\newtheorem{claim}[theorem]{Claim}
\newtheorem{definition}[theorem]{Definition}

\newif\iffinal
\finaltrue
\finalfalse

  \usepackage{nth}
  \usepackage{intcalc}

\title{Sublinear Time Low-Rank Approximation of\\Positive Semidefinite Matrices}

\author{
 Cameron Musco \\
  MIT \\
  \texttt{cnmusco@mit.edu}
\and
David P. Woodruff\\
Carnegie Mellon University \\
\texttt{dwoodruf@cs.cmu.edu }
}
\date{}

\begin{document}


\begin{titlepage}
  \maketitle
\begin{abstract}
  We show how to compute a {\it relative-error} low-rank approximation to any
  positive semidefinite (PSD) matrix in {\it sublinear time}, i.e., for any
  $n \times n$ PSD matrix $A$,  in $\tilde O(n \cdot \poly(k/\epsilon))$ time we output a rank-$k$ matrix $B$, in factored form, for which
  $\|A-B\|_F^2 \leq (1+\epsilon)\|A-A_k\|_F^2$, where $A_k$ is the best
  rank-$k$ approximation to $A$. When $k$ and $1/\epsilon$ are not too large compared to the sparsity of $A$, our algorithm does not need to read all entries of the matrix. Hence, we significantly improve upon previous
  $\nnz(A)$ time algorithms based on oblivious
  subspace embeddings, and bypass an $\nnz(A)$ time lower bound for
  general matrices (where $\nnz(A)$ denotes the number of non-zero entries in the matrix). 
  We prove time lower bounds for low-rank approximation of PSD matrices,
  showing that our algorithm is close to optimal. Finally, we extend our techniques to give sublinear time algorithms for low-rank approximation of $A$ in the (often stronger) spectral norm metric $\norm{A-B}_2^2$ and for ridge regression on PSD matrices. 
  \end{abstract}

\thispagestyle{empty}
\end{titlepage}

\section{Introduction}
%
A fundamental task in numerical linear algebra is to compute a low-rank approximation of a matrix. Such an approximation can reveal underlying low-dimensional structure, can provide a compact way of storing a matrix in factored form, and can be quickly applied to a vector. 
Countless applications include clustering \cite{dfkvv04,fss13,lbkw14,cohen2015dimensionality}, datamining \cite{afkms01}, information retrieval \cite{prtv00}, learning mixtures of distributions \cite{am05,ksv08}, recommendation systems \cite{dkr02}, topic modeling \cite{h03}, and web search \cite{afkm01,k99}. 

One of the most well-studied versions of the problem is to compute a near optimal low-rank approximation with respect to the Frobenius norm. That is, given an $n \times n$ input matrix $A$ and an accuracy parameter $\epsilon > 0$, output a rank-$k$ matrix $B$ for which:
\begin{eqnarray}\label{eqn:guarantee}
\|A-B\|_F^2 \leq (1+\epsilon) \|A-A_k\|_F^2,
\end{eqnarray}
where for a matrix $C$, $\|C\|_F^2 = \sum_{i,j} C_{i,j}^2$ is its squared Frobenius norm, and $A_k = \textrm{argmin}_{\textrm{rank-}k \textrm{ }B }\|A-B\|_F$. $A_k$ can be computed exactly using the singular value decomposition, but takes $O(n^3)$ time in practice and $n^{\omega}$ time in theory, where $\omega \approx 2.373$ is the exponent of matrix multiplication. 

In seminal work, Frieze, Kannan, and Vempala \cite{fkv04} and Achlioptas and McSherry \cite{am07} show that using randomization and approximation, much faster runtimes are possible. Specifically, \cite{fkv04} gives an algorithm that, assuming access to the row norms of $A$, outputs rank-$k$ $B$, in factored form, such that with good probability, $\|A-B\|_F^2 \leq \|A-A_k\|_F^2 + \epsilon \|A\|_F^2$. The algorithm runs in just $n \cdot \poly(k/\epsilon)$ time. However $\nnz(A)$ additional time is required to compute the row norms, where $\nnz(A)$ denotes the number of non-zero entries of $A$. Further, the guarantee achieved can be significantly weaker than (\ref{eqn:guarantee}), since the error is of the form $\epsilon\|A\|_F^2$ rather than $\epsilon \|A-A_k\|_F^2$.
Note that $\|A-A_k\|_F^2 \ll \|A\|_F^2$ precisely when $A$ is well-approximated by a rank-$k$ matrix.
Related additive error algorithms with additional assumptions were given for tensors in \cite{swz16}. 

Sarl\'os \cite{sarlos2006improved} showed how to achieve (\ref{eqn:guarantee}) with constant probability in $\tilde O(\nnz(A) \cdot k/\epsilon) + n \cdot \poly(k/\epsilon)$ time. This was improved by Clarkson and Woodruff \cite{clarkson2013low} who achieved $O(\nnz(A)) + n \cdot \poly(k/\epsilon)$ time. See also work by Bourgain, Dirksen, and Nelson \cite{bdn15}, Cohen \cite{c16}, Meng and Mahoney \cite{mm13}, and Nelson and Nguyen \cite{nn13} which further improved the degree in the $\poly(k/\epsilon)$ term. For a survey, see \cite{w14}.

In the special case that $A$ is rank-$k$ and so $\norm{A-A_k}_F^2 = 0$, \eqref{eqn:guarantee} is equivalent to the well studied low-rank matrix completion problem \cite{candes2009exact}. Much attention has focused on completing \emph{incoherent} low-rank matrices, whose singular directions are represented uniformly throughout the rows and columns and hence can be identified via uniform sampling and without fully  accessing the matrix. Under incoherence (and often condition number) assumptions, a number of methods are able to complete a rank-$k$ matrix in $\tilde O(n \cdot \poly(k))$ time \cite{jain2013low,hardt2014understanding}. 

For general matrices, without incoherence, it is not hard to see that $\Omega(\nnz(A))$ is a time lower bound: if one does not read a constant fraction of entries of $A$, with constant probability one can miss an entry much larger than all others, which needs to be included in the low-rank approximation. 

\subsection{Low-rank Approximation of Positive Semidefinite Matrices}
An important class of matrices for which low-rank approximation is often applied is the set of positive semidefinite (PSD) matrices. These are real symmetric matrices with all non-negative eigenvalues. They arise for example as covariance matrices, graph Laplacians, Gram matrices (in particular, kernel matrices), and random dot product models \cite{YS}. In multidimensional scaling, low-rank approximation of PSD matrices in the Frobenius norm error metric \eqref{eqn:guarantee} corresponds to the standard `strain minimization' problem \cite{cox2000multidimensional}. 
Completion of low-rank, or nearly low-rank (i.e., when $\norm{A-A_k}_F^2 \approx 0$), PSD matrices from few entries  is important in applications such as quantum state tomography \cite{gross2010quantum} and global positioning using local distances \cite{so2005theory,young1938discussion}.

Due to its importance, a vast literature studies low-rank approximation of PSD matrices \cite{drineas2005nystrom,zhang2008improved,KMT,belabbas2009spectral,li2010making,GM,WangZ,duan2014low,WangLZ,tropp2016randomized,li2016fast,mm16,CWPSD}. However, known algorithms either run in at least $\nnz(A)$ time,
do not achieve the relative-error guarantee of (\ref{eqn:guarantee}), or require strong incoherence assumptions\footnote{
Many of these algorithms satisfy the additional constraint that the low-rank approximation $B$ is PSD. This is also now known to be possible in $O(\nnz(A))$ time using sketching-based algorithms for general matrices \cite{CWPSD}.} (see Table \ref{resultsTable}).

At the same time, 
the simple $\Omega(\nnz(A))$ time lower bound for general matrices \emph{does not hold in the PSD case}. Positive semidefiniteness ensures that for all $i,j$, $|A_{i,j}| \leq \max(A_{i,i}, A_{j,j})$. So `hiding' a large entry in $A$ requires creating a corresponding large diagonal entry. By reading the $n$ diagonal elements, an algorithm can avoid being tricked by this approach.
While far from an algorithm, this argument raises the possibility that
improved runtimes could be possible for PSD matrices.



\subsection{Our Results}
We give the first sublinear time relative-error low-rank approximation algorithm for PSD matrices. Our algorithm reads just $n k \cdot \poly(\log n/\epsilon)$ entries of $A$ and runs in $n k^{\omega-1} \cdot \poly(\log n/\epsilon)$ time (Theorem \ref{thm:main}). With probability $99/100$ it outputs a matrix $B$ in factored form which satisfies (\ref{eqn:guarantee}). We critically exploit the intuition that large entries cannot `hide' in PSD matrices, but surprisingly require {\it no additional assumptions} on $A$, such as incoherence or bounded condition number. 


We complement our algorithm with an $\Omega(n k / \epsilon)$ time lower bound. The lower bound is information-theoretic, showing that any algorithm which reads fewer than this number of entries in the input PSD matrix cannot achieve the relative-error guarantee of (\ref{eqn:guarantee}) with constant probability. As our algorithm only reads $n k \cdot \poly(\log n/\epsilon)$ entries of $A$, this is nearly optimal for constant $\epsilon$.
We note that the actual time complexity of our algorithm is slower by a factor of $k^{\omega-2}$.

Finally, we show that our techniques can be extended to compute $B$  satisfying the spectral norm guarantee: $\norm{A-B}_2^2 \le (1+\epsilon) \norm{A-A_k}_2^2 + \frac{\epsilon}{k}  \norm{A-A_k}_F^2$ using just $nk^2 \cdot \poly(\log n/\epsilon)$ accesses to $A$ and $nk^\omega \cdot \poly(\log n/\epsilon)$ time (Theorem \ref{thm:spectral}). This guarantee is often stronger than \eqref{eqn:guarantee} when $\norm{A-A_k}_F^2$ is large, and is important in many applications. For example, we use this result to solve the ridge regression problem $\min_{x\in \mathbb{R}^n} \norm{Ax-y}_2^2 + \lambda \norm{x}_2^2$ up to $(1+\epsilon)$ relative error in $\tilde O \left  (\frac{n s_\lambda^\omega}{\epsilon^{2\omega}}  \right )$ time, where $s_\lambda = \tr((A^2 +\lambda I)^{-1} A^2)$ is the \emph{statistical dimension} of the problem (see Theorem \ref{thm:regression}). Typically $s_\lambda \ll n$, so our runtime is sublinear and improves significantly on existing input-sparsity time results \cite{HCW}. 
For a summary of our results and comparison to prior work see, Table \ref{resultsTable}.

\begin{table}
\begin{center}
  \begin{tabular}{| c | c | c |}
    \hline
    \textbf{Source} & \textbf{Runtime} &\textbf{Approximation Bound} \\ \hline
    \cite{drineas2005nystrom}  & $n k^{\omega-1} \cdot \poly(1/\epsilon)$& $\norm{A-B}_F \le \norm{A-A_k}_F + \epsilon \norm{A}_*$\tablefootnote{Note that this bound is stated incorrectly as $\norm{A-B}_F \le \norm{A-A_k}_F + \epsilon \sum_{i=1}^n (A_{ii})^2$ in \cite{drineas2005nystrom}.} \\ \hline
    \cite{KMT} & $nk^{\omega-1} \cdot \poly(1/\epsilon)$ & $\norm{A-B}_F \le \norm{A-A_k}_F + \epsilon n \cdot \max_i A_{ii}$
    \\ \hline
     \cite{GM} & $\tilde O(n^2) + nk^{\omega-1}\poly(\log n/\epsilon)$ & $\norm{A-B}_F \le \norm{A-A_k}_F + \epsilon \norm{A-A_k}_*$\\ \hline
    \cite{agr16} + \cite{belabbas2009spectral} & $n \log n \cdot \poly(k)$ & $\norm{A-B}_F \le (k+1) \norm{A-A_k}_*$\\ \hline
    \cite{mm16} & $n k^{w-1} \cdot \poly(\log k/\epsilon)$ & $\norm{\ha-B}_F^2 \le (1+\epsilon) \norm{\ha-\ha_k}_F^2$ \\
    \hline 
    \cite{CWPSD} & $O(\nnz(A)) + n \poly(k/\epsilon)$ & $\norm{A-B}_F^2 \le (1+\epsilon) \norm{A-A_k}_F^2$ \\
    \hline
    \multicolumn{3}{|c|}{\textbf{Our Results}}\\
    \hline
    Theorem \ref{thm:main} & $n k^{\omega-1} \cdot \poly(\log n/\epsilon)$ & $\norm{A-B}_F^2 \le (1+\epsilon) \norm{A-A_k}_F^2$ \\ \hline
    Theorem \ref{thm:spectral} & $nk^\omega \cdot \poly(\log n/\epsilon)$ & $\norm{A-B}_2^2 \le (1+\epsilon) \norm{A-A_k}_2^2 + \frac{\epsilon}{k}  \norm{A-A_k}_F^2$ \\\hline
  \end{tabular}
\end{center}
\caption{Comparison of our results to prior work on low-rank approximation of PSD matrices. $\norm{M}_* = \sum_{i=1}^n \sigma_i(M)$ denotes the nuclear norm of matrix $M$.
The cited results all output $B$ (in factored form), which is itself PSD. In Theorem \ref{thm:psdOutput} we show how to modify our algorithm to satisfy this condition, and run in $nk^\omega \cdot \poly(\log n/\epsilon)$ time. The table shows results that do not require incoherence assumptions on $A$. For general PSD matrices, all known incoherence based results (see e.g., \cite{gittens2011spectral,GM}) degrade to $\Theta(n^\omega)$ runtime. 
Additionally, as discussed, any general low-rank approximation algorithm can be applied to PSD matrices, with state-of-the-art approaches running in input-sparsity time \cite{clarkson2013low,mm13,nn13}. \cite{CWPSD} extends these results to the case where the output $B$ is restricted to be PSD. \cite{tropp2016randomized} does the same but outputs $B$ with rank $2k$.
For a more in depth discussion of bounds obtained in prior work, see Section \ref{sec:intuition}.
} 
\label{resultsTable}
\end{table}
\vspace{-.25em}

\subsection{Algorithm Overview} 

The starting point for our approach is the fundamental fact that \emph{any} matrix $A$ contains a subset of $O(k/\epsilon)$ columns, call them $C$, that span a relative-error rank-$k$ approximation to $A$ \cite{deshpande2006matrix,deshpande2006adaptive,drineas2006subspace}. Computing the best low-rank approximation to $A$ using an SVD requires access to all $\Theta(n^2)$ dot products between the columns of the matrix. However, given $C$, just  $n \cdot O(k/\epsilon)$ dot products are needed -- to project the remaining columns of the matrix to the span of the subset. 

Additionally, a subset of size $\poly(k/\epsilon)$ can be identified using an intuitive approach known as \emph{adaptive sampling} \cite{deshpande2006adaptive}: columns are iteratively added to the subset, with each new column being sampled with probability proportional to its norm \emph{outside the column span} of the current subset. Formally, column $a_i$ is selected with probability $\frac{\norm{a_i - P_Ca_i}_2^2}{\norm{A- P_CA}_F^2}$ where $P_C$ is the projection onto the current subset $C$. Computing these sampling probabilities requires knowing the norm of each $a_i$ along with its dot product with each column currently in $C$. So, overall this approach gives a relative-error low-rank approximation using just $n \cdot \poly(k/\epsilon)$ dot products between columns of $A$.

The above observation is surprising -- not only does every matrix contain a small column subset witnessing a near optimal low-rank approximation, but also, such a witness can be found using significantly less information about the column span of the matrix than is required by a full SVD. 

This fact is not immediately algorithmically useful, as computing the required dot products takes $\nnz(A) \cdot \poly(k/\epsilon)$ time. However, given PSD $A$, we can write the eigendecomposition $A = U \Lambda U^T$ where $\Lambda$ is a non-negative diagonal matrix of eigenvalues, and let $A^{1/2} = U \Lambda^{1/2} U^T$ be the matrix square root of $A$. Since $\ha\ha = A$, the entry $A_{i,j}$ is just the dot product between the $i^{th}$ and $j^{th}$ columns of $\ha$. So with $A$ in hand, the dot products have been `precomputed' and the above approach yields a low-rank approximation algorithm for $\ha$ running in just $n \cdot \poly(k/\epsilon)$ time. Note that, aligning with our initial intuition that reading the diagonal entries of $A$ is necessary to avoid the $\nnz(A)$ time lower bound for general matrices, the diagonal entries of $A$ are the column norms of $\ha$, and hence their values are critical to computing the adaptive sampling probabilities.

By the above argument, given PSD $A$, we can compute in $n \cdot \poly(k/\epsilon)$ time a rank-$k$ orthogonal projection matrix $P \in \mathbb{R}^{n \times n}$ (in factored form) for which $\|A^{1/2} - A^{1/2}P\|_F^2 \leq (1+\epsilon)\|A^{1/2} - A^{1/2}_k\|_F^2$. This approach can be implemented using adaptive sampling \cite{deshpande2006adaptive}, sublinear time volume sampling \cite{agr16}, or as shown in \cite{mm16}, recursive \emph{ridge leverage score sampling}. The ridge leverage scores are a natural interpolation between adaptive sampling and the widely studied leverage scores, which, as we will see, have a number of additional algorithmically useful properties.
As discussed in \cite{mm16}, the guarantee for $\ha$ is useful for a number of kernel learning methods such as kernel ridge regression. However, it is very different from our final goal. 
%
%
%
In fact, one can show that projecting to $P$ can yield an {\it arbitrarily bad} low-rank approximation to $A$ itself (\iffinal\else see Appendix \ref{sec:example}\fi).

We note that, since $P$ is constructed via column selection methods, it is possible to efficiently compute a factorization of $\ha P \ha$ (\iffinal\else see Appendix \ref{sec:example}\fi). Further, this matrix gives a near optimal low-rank approximation of $A$ if we use error parameter $\epsilon' = \epsilon/\sqrt{n}$. This approach gives a first sublinear time algorithm, but it is significantly suboptimal. Namely, it requires reading $\tilde O(n k/\epsilon') = \tilde O(n^{3/2} k /\epsilon)$ entries of $A$ and takes $n^{1.69} \cdot \poly(k/\epsilon)$ time using fast matrix multiplication. 

To improve the dependence on $n$, we need a better understanding of how to perform ridge leverage score sampling on $A$ itself. 
We start by showing that the ridge leverage scores of $A^{1/2}$ are within a factor of $O(\sqrt{n/k})$ of the ridge leverage scores of $A$. By this bound, if we over-sample columns of $A$ by a factor of $O(\sqrt{n/k})$ using the ridge leverage scores of $A^{1/2}$ (computable via \cite{mm16}), obtaining a sample of $\tilde O(\sqrt{n/k} \cdot k/\epsilon^2)$ columns, the sample will be a so-called {\it projection-cost preserving sketch} (PCP) of $A$. The notion of a PCP was introduced in \cite{cohen2015dimensionality}, where a matrix $C$ is defined to be an $(\epsilon, k)$-column PCP of $A$ if for all rank-$k$ projection matrices $P$:
\begin{eqnarray}\label{eqn:pcp}
(1-\epsilon) \|A-PA\|_F^2 \leq \|C - PC\|_F^2 \leq (1+\epsilon)\|A-PA\|_F^2.
\end{eqnarray}
One important property of a PCP is that good low-rank approximations to $C$ translate to good low-rank approximations of $A$. More precisely, if $U$ is an $n \times k$ matrix with orthonormal columns for which $\|C-UU^TC\|_F^2 \leq (1+\epsilon)\|C-C_k\|_F^2$, then $\|A-UU^TA\|_F^2 \leq \frac{(1+\epsilon)^2}{(1-\epsilon)}\|A-A_k\|_F^2$. 

Letting $C$ be the $n \times \tilde O(\sqrt{nk}/\epsilon^2)$ submatrix which we sample via ridge leverage scores, we can apply an $\nnz(C)$ time algorithm to compute a subspace $U \in \R^{n \times k}$ whose columns span a near-optimal low-rank approximation of $C$, and hence of $A$ by the PCP property. Using standard sampling techniques, we can approximately project the columns of $A$ to $U$, producing our final solution. This gives time complexity $n^{3/2} \cdot \poly(k/\epsilon)$, improving slightly upon our first approach. 


To reduce the time to linear in $n$, we must further reduce the size of $C$ by sampling
a small subset of its rows, which themselves form a PCP.
To find these rows, we cannot
afford to use oblivious sketching techniques, which would take at least $\nnz(C)$ time, nor
can we use our previous method for providing $O(\sqrt{n/k})$ overestimates
to the ridge leverage scores, since $C$ is no longer PSD. In fact, the row ridge leverage scores of $C$ can be arbitrarily large compared to those of $\ha$. 

The key idea to getting around this issue is that, since $C$ is a column PCP of $A$, projecting its columns onto $A$'s top eigenvectors gives a near optimal low-rank approximation. Further, we can show that the ridge leverage scores of $\ha$ (appropriately scaled) upper bound the \emph{standard leverage scores} of this low-rank approximation. Sampling by these leverage scores is not enough to give a guarantee like \eqref{eqn:pcp} -- they ignore the entire component of $C$ not falling in the span of $A$'s top eigenvectors and so may significantly distort projection costs over the matrix.
Further, at this point, we have no idea how to estimate the row norms of $C$, or even its Frobenius norm, with $n \poly(k/\epsilon)$ samples, which are necessary to implement any kind of adaptive sampling approach.

Fortunately, using that sampling at least preserves the matrix in expectation, along with a few other properties of the ridge leverage scores of $\ha$, we show that, with good probability, sampling $\tilde O(\sqrt{nk}/\poly(\epsilon))$ rows of $C$ by these scores yields $R$ satisfying for all rank-$k$ projection matrices $P$: 
$$(1-\epsilon)\|C-CP\|_F^2 \leq \|R-RP\|_F^2 + \Delta \leq (1+\epsilon)\|C-CP\|_F^2$$
where $\Delta$ is a fixed value, independent of $P$, with $|\Delta| \le c\norm{C-C_k}_F^2$ for some constant $c$. Since the same $\Delta$ distortion applies to all $P$, and since it is at most a constant times the true optimum, a near optimal low-rank approximation for $R$ still translates to a near optimal approximation for $C$.  

At this point $R$ is a small matrix, and we can
run any $O(\nnz(R))$ time algorithm to find a good low-rank factorization $EF^T$ to it, where $F^T$
is $k \times \tilde O(\sqrt{nk}/\poly(\epsilon))$. Since $R$ is a row PCP for $C$, 
by regressing the rows of $C$ to the span of $F$, we can obtain a near optimal low-rank approximation to $C$.
We can solve this multi-response regression approximately in sublinear time via standard sampling techniques. Approximately regressing $A$ to the span of this approximation using similar techniques yields our final result.
The total runtime is dominated by the input-sparsity low-rank approximation of $R$ requiring $O(\nnz(R)) = \tilde O(nk /\poly(\epsilon))$ time.

To improve $\epsilon$ dependencies in our final runtime, achieving sample complexity $\tilde O \left (\frac{nk}{\epsilon^{2.5}} \right)$, we modify  this approach somewhat, showing that $R$ actually satisfies a stronger \emph{spectral norm PCP} property for $C$. This property lets us find a low-rank span $Z$ with $\norm{C-CZZ^T}_2^2 \le \frac{\epsilon}{k} \norm{A-A_k}_F^2$, from which, through a series of approximate regression steps, we can extract a low-rank approximation to $A$ satisfying \eqref{eqn:guarantee}. This stronger spectral guarantee also lies at the core of our extensions to near optimal spectral norm low-rank approximation (Theorem \ref{thm:spectral}), ridge regression (Theorem \ref{thm:regression}), and low-rank approximation where $B$ is restricted to be PSD (Theorem \ref{thm:psdOutput}).
\subsection{Some Further Intuition on Error Guarantees}\label{sec:intuition}

Observe that in computing a low-rank approximation of $A$, we read just $\tilde O(n \cdot \poly(k/\epsilon))$ entries of the matrix, which is, up to lower order terms, the same number of entries (corresponding to column dot products of $\ha$) that we accessed to compute a low-rank approximation of $\ha$ in our description above. However, these sets of entries are very different. While low-rank approximation of $\ha$ looks at an $n \times \poly(k/\epsilon)$ sized submatrix of $A$ together with the diagonal entries, our algorithm considers a carefully chosen $\sqrt{nk} \poly(\log n/\epsilon) \times \sqrt{nk} \poly(\log n/\epsilon)$ submatrix together with the diagonal entries, which gives significantly more information about the spectrum of $A$. 

As a simple example, consider $A$ with top eigenvalue $\lambda_1 = \sqrt{n}$, and $\lambda_i = 1$ for $i=2,...n$. $\norm{\ha}_F^2 = \sum_{i=1}^n \lambda_i = \sqrt{n} + n -1$ while $\norm{\ha - \ha_1}_F^2 = \sum_{i=2}^n \lambda_i = n-1$. So, $\ha$ has no good rank-$1$ approximation. Unless we set $\epsilon = O(1/\sqrt{n})$, a low-rank approximation algorithm for $\ha$ can learn nothing about $\lambda_1$ and still be near optimal. In contrast, $\norm{A}_F^2 = \sum_{i=1}^n \lambda_i^2 = 2n -1$ and $\norm{A-A_1}_F^2 = \sum_{i=2}^n \lambda_i^2 = n-1$. So, even with $\epsilon = 1/2$, any rank-$1$ approximation algorithm for $A$ must identify the presence of $\lambda_1$ and project this direction off the matrix.
In this sense, our algorithm is able to obtain a much more accurate picture of $A$'s spectrum. 

With incoherence assumptions, prior work on PSD low-rank approximation \cite{GM} obtains the bound $\norm{A-B}_* \le (1+\epsilon) \norm{A-A_k}_*$ in sublinear time, where $\norm{M}_* = \sum_{i=1}^n \sigma_i(M)$ is the nuclear norm of $M$. Recent work (\cite{agr16} in combination with \cite{belabbas2009spectral}) gives $\norm{A-B}_F \le (k+1) \norm{A-A_k}_*$ without the incoherence assumption. These nuclear norm bounds are closely related to approximation bounds for $\ha$ and it is not hard to see that neither require $\lambda_1$ to be detected in the example above, and so in this sense are weaker than our Frobenius norm bound.

A natural question if even stronger bounds are possible: e.g., can we compute $B$ with $\norm{A-B}_2^2 \le (1+\epsilon) \norm{A-A_k}_2^2$ in sublinear time? We partially answer this question in Theorem \ref{thm:spectral}. In $\tilde O(nk^\omega \poly(\log n/\epsilon))$ time, we can find $B$ satisfying $\norm{A-B}_2^2 \le (1+\epsilon) \norm{A-A_k}_2^2 + \frac{\epsilon}{k}  \norm{A-A_k}_F^2$.

Significantly improving the above bound seems difficult: it is easy to see that a relative error spectral norm guarantee requires $\Omega(n^2)$ time. Consider $A$ which is the identity except with $A_{i,j} = A_{j,i} = 1$ for some uniform random pair $(i,j)$. Finding $(i,j)$ requires $\Omega(n^2)$ queries to $A$. However, it is necessary to achieve a  relative error  spectral norm guarantee with $\epsilon < 3$ since $\norm{A}_2^2 = 4$ while $\norm{A-A_1}_2^2 = 1$ where $A_1$ is all zeros with ones at its $(i,i)$, $(j,j)$, $(i,j)$, and $(j,i)$ entries. 

A similar argument shows that relative error low-rank approximation in higher Schatten-$p$ norms, i.e., $\norm{A-B}_p^p$ for $p > 2$ requires superlinear dependence on $n$ (where $\norm{M}_p^p = \sum_{i=1}^n \sigma_i^p(M)$.) We can set $A$ to be the identity but with an all ones block on a uniform random subset of $n^{1/p}$ indices. This block has associated eigenvalue $\lambda_1 = n^{1/p}$ and so, since all other ($n-n^{1/p}$) eigenvalues of $A$ are $1$, $\norm{A}_p^p = \Theta(n)$, and the block must be recovered to give a relative error approximation to $\norm{A-A_1}_p^p$. However, as the block is placed uniformly at random and contains just $n^{2/p}$ entries, finding even a single entry requires $n^{2-2/p}$ queries to $A$ -- superlinear for $p > 2$.

\subsection{Open Questions}
While it is apparent that obtaining stronger error guarantees than \eqref{eqn:guarantee} may require increased runtime, understanding exactly what can be achieved in sublinear time is an interesting direction for future work. We also note that it is still unknown how to compute a number of  basic properties of PSD matrices in sublinear time. For example, while we can output $B$ satisfying $\norm{A-B}_F^2 \le (1+\epsilon) \norm{A-A_k}_F^2$, surprisingly it is not clear how to actually  estimate the value $\norm{A-A_k}_F^2$ to within a $(1\pm \epsilon)$ factor. This can be achieved in $n^{3/2} \poly(k/\epsilon)$ time using our PCP techniques. However, obtaining linear runtime in $n$ is open. Estimating $ \norm{A-A_k}_F^2$ seems strongly connected to estimating other important quantities such as the statistical dimension of $A$ for ridge regression (see Theorem \ref{thm:regression}) which we do not know how to do in $o(n^{3/2})$ time.

Finally, an open question is if these techniques can be generalized to a broader class of matrices. As discussed, in the matrix completion literature, much attention has focused on incoherent low-rank matrices \cite{candes2009exact} which can be approximated with uniform sampling. PSD matrices are not incoherent in general, which is highlighted by the fact that our sampling schemes are far from uniform and very adaptive to previously seen matrix entries. However, perhaps there is some other parameter (maybe relating to a measure of diagonal dominance) which characterizes when low-rank approximation can be performed with just a small number of adaptive accesses to $A$.

\subsection{Paper Outline}

\noindent \textbf{Section \ref{sec:ridge}: Ridge Leverage Score Sampling}. We show that the ridge leverage scores of $A$ are within an $O(\sqrt{n/k})$ factor of those of $A^{1/2}$, letting us use the fast ridge leverage score sampling algorithm of \cite{mm16} to sample $\tilde O(\sqrt{nk}/\epsilon^2)$ columns of $A$ that form a column PCP of the matrix.

\medskip
\noindent \textbf{Section \ref{sec:row}: Row Sampling}.
We discuss how to further accelerate our algorithm by obtaining a row PCP for our column sample, allowing us to achieve runtime linear in $n$.

\medskip
\noindent \textbf{Section \ref{sec:full}: Full Algorithm}.
We use the primitives in the previous sections along with approximate regression techniques to give our full sublinear time low-rank approximation algorithm.

\medskip
\noindent \textbf{Section \ref{sec:lower}: Lower Bounds}.
We show that our algorithm is nearly optimal -- any relative error low-rank approximation algorithm must read $\Omega(nk/\epsilon)$ entries of $A$. 

\medskip
\noindent \textbf{Section \ref{sec:spectral}: Spectral Norm Bounds}. We modify the algorithm of Section \ref{sec:full} to give a tighter approximation in the spectral norm and discuss applications to sublinear time ridge regression.

\section{Ridge Leverage Score Sampling}\label{sec:ridge}
Our main
algorithmic tool will be ridge leverage score sampling, which is used to identify a small subset of columns of $A$ that span a good low-rank approximation of the matrix. Following the definition of \cite{cohen2015ridge}, the rank-$k$ ridge leverage scores of any matrix $A$ are given by:
\begin{definition}[Ridge Leverage Scores]\label{def:ridgeScores} For any $A \in \mathbb{R}^{n \times d}$, letting $a_i \in \mathbb{R}^{n}$ be the $i^{th}$ column of $A$, the $i^{th}$ rank-$k$ column ridge leverage score of $A$ is:
\begin{align*}
\tau_i^k(A) = a_i^T \left (AA^T + \frac{\norm{ A-A_k}_F^2}{k} I\right )^{+} a_i.
\end{align*}
\end{definition}
Above $I$ is the appropriately sized identity matrix and $M^+$ denotes the matrix pseudoinverse, equivalent to the inverse unless $\norm{ A-A_k}_F^2 = 0$ and $A$ is singular. Analogous scores can be defined for the rows of $A$ by simply transposing the matrix.
It is not hard to see that $0 < \tau^k_i(A) < 1$ for all $i$. Since we use these scores as sampling probabilities, it is critical that the sum of scores, and hence the size of the subsets we sample, is not too large.
We have the following
(\iffinal\else see Appendix \ref{additional_proofs}\fi):
\begin{lemma}[Sum of Ridge Leverage Scores]\label{lem:sum}
For any $A \in \R^{n\times d}$, $\sum_{i=1}^d \tau^k_i(A) \le 2k$.
\end{lemma}

Intuitively, the ridge leverage scores are similar to the standard leverage scores of $A$, which are given by $a_i^T (AA^T)^+ a_i$. By writing $A = U \Sigma V^T$ in its SVD, one sees that standard leverage scores
are just the squared column norms of $V^T$.
Sampling columns by ridge leverage scores yields a spectral approximation to the matrix. The addition of the weighted identity (or `ridge') $\frac{\norm{ A-A_k}_F^2}{k} I$  `dampens' contributions from smaller singular directions of $A$, decreasing the sum of the scores and allowing us to sample fewer columns. At the same time, it introduces error dependent on the size of the tail $\norm{A-A_k}_F^2$, ultimately giving an approximation from which it is possible to output a near optimal low-rank approximation to the original matrix. Specifically, sampling by ridge leverage scores yields a \emph{projection-cost preserving} sketch (PCP) of $A$:

\begin{lemma}[Theorem 6 of \cite{cohen2015ridge}]\label{thm:pcp} For any $A \in \R^{n \times d}$, for $i \in \{1,\ldots,d\}$, let $\tilde \tau_i^k \ge \tau_i^k(A)$ be an overestimate for the $i^{th}$ rank-$k$ ridge leverage score. 
Let $p_i = \frac{\tilde \tau^k_i}{\sum_i \tilde \tau^k_i}$ and $t = \frac{c\log(k/\delta)}{\epsilon^2} \sum_i \tilde \tau^k_i$ for any $\epsilon < 1$ and sufficiently large constant $c$. Construct $C$ by sampling $t$ columns of $A$, each set to $\frac{1}{\sqrt{tp_i}}a_i$ with probability $p_i$.
With probability $1-\delta$, for any rank-$k$ orthogonal projection $P \in \R^{n \times n}$,
\begin{align*}
(1-\epsilon)\|A - PA\|^2_F \leq \|C - PC\|^2_F\leq (1+\epsilon)\|A - PA\|^2_F.
\end{align*}
We refer to $C$ as an $(\epsilon,k)$-column PCP of $A$.
\end{lemma}

Since the 	`cost' $\norm{A-PA}_F^2$ of any rank-$k$ projection of $A$ is preserved by $C$, 
any near-optimal low-rank approximation of $C$ yields a near optimal low-rank approximation of $A$. Further, $C$ is much smaller than $A$, so such a low-rank approximation can be computed quickly.
%
The difficulty is in computing the approximate leverage scores. To do this,
we use the main result from \cite{mm16}:
\begin{lemma}[Corollary of Theorem 20 of \cite{mm16}]\label{thm:originalSampling}
There is an algorithm that given any PSD matrix $A \in \R^{n\times n}$, runs in $O(n(k\log (k/\delta))^{\omega-1})$ time, accesses $O(nk \log (k/\delta))$ entries of $A$, and
 returns for each $i \in [1,..,n]$, $\tilde \tau^k_i(A^{1/2})$ such that with probability $1-\delta$, for all $i$:
\begin{align*}
\tau^k_i(A^{1/2}) \le \tilde \tau^k_i(A^{1/2}) \le 3\tau^k_i(A^{1/2}).
\end{align*}
\end{lemma}
\begin{proof}
Theorem 20 of \cite{mm16} shows that by using a recursive ridge leverage score sampling algorithm, it is possible to return (with probability $1-\delta$) a sampling matrix $S \in \mathbb{R}^{n \times s}$ with $s = O \left (k\log(k/ \delta) \right )$ such that, letting $\lambda = \frac{1}{k}\norm{\ha-\ha_k}_F^2$:
\begin{align*}
\frac{1}{2} \left ( A + \lambda I \right ) \preceq \left (\ha SS^T \ha + \lambda I \right ) \preceq \frac{3}{2} \left (A + \lambda I \right)
\end{align*}
where $M \preceq N$ indicates $x^TMx \le x^TN x$ for all $x$. If
we set $\tilde \tau^k_i(\ha) = 2 \cdot x_i^T \left (\ha SS^T \ha + \lambda I \right )^{+}  x_i$, where $x_i$ is the $i^{th}$ column of $A^{1/2}$ we have the desired bound. Of course, we cannot directly compute this value without factoring $A$ to form $A^{1/2}$. However, as shown in Lemma 6 of \cite{mm16}:
\begin{align*}
x_i^T \left (\ha SS^T \ha + \lambda I \right )^{-1}  x_i = \frac{1}{\lambda} \left ( A - AS(S^TAS + \lambda I)^{-1} S^TA \right )_{i,i}.
\end{align*}

Computing $(S^TAS + \lambda I)^{-1}$ requires $O(s^2) = O((k\log (k/\delta))^2)$ accesses to $A$ and $O(s^\omega) = O((k\log (k/\delta))^\omega)$ time. Computing all $n$ diagonal entries of $AS(S^TAS + \lambda I)^{-1} S^TA$ then requires $O(nk \log (k/\delta))$ accesses to $A$ and $ O(n(k\log (k/\delta))^{\omega-1})$ time.
With these entries in hand we can simply subtract from the diagonal entries of $A$ and rescale to give the final leverage score approximation. Critically, this calculation always reads \emph{all diagonal entries} of $A$, allowing it to identify rows containing large off diagonal entries and skirt the $\nnz(A)$ time lower bound for general matrices.

Note that the stated runtime in \cite{mm16} for outputting $S$ is $\tilde O(n k)$ accesses to $A$ (\emph{kernel evaluations} in the language of \cite{mm16}) and $\tilde O(nk^2)$ runtime. However this runtime is improved to $\tilde O(nk^{\omega-1})$ using fast matrix multiplication. 
\end{proof}

In order to apply Lemmas \ref{thm:pcp} and \ref{thm:originalSampling} to low-rank approximation of $A$, we now show that the ridge leverage scores of $\ha$ coarsely approximate those of $A$:

\begin{lemma}[Ridge Leverage Score Bound]\label{lem:scoreBound} For any PSD matrix $A \in \mathbb{R}^{n\times n}$:
\begin{align*}
\tau_i^k(A) \le 2\sqrt{\frac{n}{k}} \cdot \tau_i^k(\ha).
\end{align*}
\end{lemma}
\begin{proof}
We write $\ha$ in its eigendecomposition $\ha = U \Lambda^{1/2} U^T$, where $\Lambda_{i,i} = \lambda_i$ is the $i^{th}$ eigenvalue of $A$. Letting $x_i$ denote the $i^{th}$ column of $\ha$ we have:
\begin{align*}
\tau_i^k (\ha ) &= x_i^T \left (A + \frac{\norm{\ha-\ha_k}_F^2}{k} I \right )^{-1} x_i = x_i^T U \bar \Lambda U^T x_i
\end{align*} 
where
$
\bar \Lambda_{i,i} \eqdef \frac{1}{\lambda_i + \frac{1}{k}\sum_{j=k+1}^n \lambda_j}.
$
We can similarly write:
\begin{align*}
\tau_i^k(A) &= a_i^T \left (A^2 + \frac{\norm{A-A_k}_F^2}{k} I \right )^{-1} a_i\\
&= x_i^T \ha \left (A^2 + \frac{\norm{A-A_k}_F^2}{k} I \right )^{-1} \ha x_i\\
& = x_i^T U \hat \Lambda U^T x_i
\end{align*} 
where 
$
\hat \Lambda_{i,i} \eqdef \frac{\lambda_i}{\lambda_i^2 + \frac{1}{k}\sum_{j=k+1}^n \lambda_j^2}.
$
Showing 
$\hat \Lambda \preceq 2\sqrt{\frac{n}{k}} \cdot \bar \Lambda$ is enough to give the lemma. Specifically we must show, for all $i$, $\hat \Lambda_{i,i} \le  2\sqrt{\frac{n}{k}} \cdot \bar \Lambda_{i,i}$ which after cross-multiplying is equivalent to:
\begin{align}\label{finalGoal}
 \lambda_i^2 + \frac{1}{k}\lambda_i \sum_{j=k+1}^n \lambda_j \le 2\sqrt{\frac{n}{k}} \left( \lambda_i^2 + \frac{1}{k}\sum_{j=k+1}^n \lambda_j^2 \right ).
\end{align}
First consider the relatively large eigenvalues. Say we have $\frac{1}{k} \sum_{j=k+1}^n \lambda_j \le \sqrt{\frac{n}{k}} \lambda_i$. Then:
\begin{align*}
 \lambda_i^2 + \frac{1}{k}\lambda_i \sum_{j=k+1}^n \lambda_j &\le \left (1 + \sqrt{n/k} \right )\lambda_i^2
\end{align*}
which gives \eqref{finalGoal}.
Next consider small eigenvalues with $ \frac{1}{k}\sum_{j=k+1}^n \lambda_j \ge \sqrt{\frac{n}{k}} \lambda_i$. In this case:
\begin{align*}
 \lambda_i^2 + \frac{1}{k}\lambda_i \sum_{j=k+1}^n \lambda_j &\le \lambda_i^2 + \frac{1}{\sqrt{n}\cdot k^{3/2}} \left (\sum_{j=k+1}^n \lambda_j \right )^2\\
&\le \lambda_i^2 + \frac{1}{\sqrt{n} \cdot k^{3/2}} \cdot n \sum_{j=k+1}^n \lambda_j^2\tag{\text{Norm bound: }$\norm{\cdot}_1^2 \le n\norm{\cdot}_2^2$}\\
&\le \sqrt{\frac{n}{k}} \left (\lambda_i^2 + \frac{1}{k}\sum_{j=k+1}^n \lambda_j^2 \right )
\end{align*}
which gives \eqref{finalGoal}, completing the proof.
\end{proof}

Combining Lemmas \ref{lem:sum}, \ref{thm:pcp}, \ref{thm:originalSampling}, \ref{lem:scoreBound} we have:
\begin{corollary}[Fast PSD Ridge Leverage Score Sampling]\label{cor:samplingCor}
There is an algorithm that given any PSD matrix $A \in \mathbb{R}^{n\times n}$ runs in $\tilde O(n k^{\omega-1})$ time, accesses $\tilde O(nk)$ entries of $A$, and with prob. $1-\delta$ outputs a weighted sampling matrix $S_1 \in \R^{n \times \tilde O \left (\frac{\sqrt{nk}}{\epsilon^2} \right )}$ such that $AS_1$ is an $(\epsilon$,$k)$-column PCP of $A$.
\end{corollary}
\begin{proof}
By Lemma \ref{thm:originalSampling} we can compute constant factor approximations to the ridge leverage scores of $\ha$ in time $\tilde O(nk^{\omega-1})$. Applying Lemma \ref{lem:scoreBound}, if we scale these scores up by $2\sqrt{n/k}$ they will be overestimates of the ridge leverage scores of $A$.
If we set $t = O \left (\frac{\log(k/\delta)}{\epsilon^2} \cdot \sum \tilde \tau_i^k \right )$, and generate $S_1$ by sampling $t$ columns of $A$ with probabilities proportional to these estimated scores, by Lemma \ref{thm:pcp}, $AS_1$ will be an $(\epsilon,k)$-column PCP of $A$ with probability $1-\delta$.
By Lemma \ref{lem:sum}, $\sum_{i=1}^n \tau_i^k(\ha) \le 2k$. So we have $t = \tilde O(\sum \tilde \tau_i^k/\epsilon^2) = \tilde O(\sqrt{nk}/\epsilon^2)$. 
\end{proof}

Forming $AS_1$ requires reading just $\tilde O(n^{3/2}\sqrt{k}/\epsilon^2)$ entries of $A$. At this point, we could 
employ any input sparsity time algorithm to find a near optimal rank-$k$ projection $P$ for approximating $AS_1$ in $O(\nnz(AS_1)) + n\poly(k/\epsilon) = n^{3/2} \cdot \poly(k/\epsilon)$ time. 
This would in turn yield a near optimal low-rank approximation of $A$. 
However, as we will see in the next section, by further sampling the rows of $AS_1$, we can significantly improve this runtime.

\section{Row Sampling}\label{sec:row}

To achieve near linear dependence on $n$, we sample roughly $\sqrt{nk}$ rows from $AS_1$, producing an even smaller matrix $S_2^T A S_1$, which we can afford to fully read and from which we can form a near optimal low-rank approximation to $AS_1$ and consequently to $A$. 
However, sampling $AS_1$ is challenging: we cannot employ input sparsity time methods as we cannot afford to read the full matrix, and since it is no longer PSD, we cannot apply the same approach we used for $A$, approximating the ridge leverage scores with those of $\ha$. 

Rewriting Definition \ref{def:ridgeScores} using the SVD $AS_1 = U\Sigma V^T$ (and transposing $AS_1$ to give row instead of column scores) we see that the row ridge leverage scores of $AS_1$ are the diagonal entries of: 
$$AS_1 \left (S_1^T A^T A S_1 + \frac{\norm{AS_1-(AS_1)_k}_F^2}{k} I\right)^+ S_1^TA^T = U \bar \Sigma U^T$$ 
where $\bar \Sigma_{i,i} = \frac{\Sigma_{i,i}^2}{\Sigma_{i,i}^2 + \frac{\norm{AS_1-(AS_1)_k}_F^2}{k}}$. That is, the row ridge leverage scores depend only on the column span $U$ of $AS_1$ and its spectrum. Since $AS_1$ is a column PCP of $A$ this gives hope that the two matrices have similar row ridge leverage scores. 

Unfortunately, this is \emph{not the case}. It is possible to have rows in $AS_1$ with ridge leverage scores significantly higher than in $A$. Thus, even if we knew the ridge leverage scores of $A$, we would have to scale them up significantly to sample from $AS_1$. As an example,
consider $A$ with relatively uniform ridge leverage scores: $\tau_i(A) \approx k/n$ for all $i$. When a column is selected to be included in $AS_1$ it will be reweighted by roughly a factor of $\sqrt{n/k}$. Now, append a number of rows to $A$ each with very small norm and just a containing single non-zero entry. These rows will have little effect on the ridge leverage scores if their norms are small enough. However, if the column corresponding to the nonzero in a row is selected, the row will appear in $AS_1$ with $\sqrt{n/k}$ times the weight that it appears in $A$, and its ridge leverage score will be roughly a factor $n/k$ times higher. 

Fortunately, we are still able to show that sampling the rows of $AS_1$ by the rank $k' = O(k/\epsilon)$ leverage scores of $\ha$ scaled up by a $\sqrt{n/k'}$ factor yields a row PCP for this matrix. Our proof works not with the ridge scores of $AS_1$ but with  
the \emph{standard leverage scores} of a near optimal low-rank approximation to this matrix -- specifically the approximation given by projecting onto the top eigenvectors of $A$. We have:

\begin{lemma}[Row PCP]\label{leveragePCP} For any PSD $A \in \R^{n \times n}$ and $\epsilon \le 1$ let $k' = \lceil c k/\epsilon \rceil$ and let $\tilde \tau_i^{k'}(\ha) \ge \tau_i^{k'}(\ha)$ be an overestimate for the $i^{th}$ rank-$k'$ ridge leverage score of $\ha$. Let $\tilde \ell_i = \sqrt{\frac{16n\epsilon}{k}} \cdot \tilde \tau_i^{k'}(\ha) $, $p_i = \frac{\tilde \ell_i}{\sum_i \tilde \ell_i}$, and $t = \frac{c'\log n}{\epsilon^2} \sum_i \tilde \ell_i$.
Construct weighted sampling matrices $S_1,S_2 \in \R^{n \times t}$ each whose $j^{th}$ column is set to $\frac{1}{\sqrt{tp_i}} e_i$ with probability $p_i$.

For sufficiently large constants $c,c'$, with probability $\frac{99}{100}$, letting $\tilde A = S_2^T A S_1$, for any rank-$k$ orthogonal projection $P \in \mathbb{R}^{t \times t}$:
\begin{align*}
(1-\epsilon) \norm{AS_1(I-P)}_F^2 \le \norm{\tilde A(I-P)}_F^2 + \Delta \le (1+\epsilon) \norm{AS_1(I-P)}_F^2
\end{align*}
 for some fixed $\Delta$ (independent of $P$) with $|\Delta| \le 600\norm{A-A_k}_F^2$. We refer to $\tilde A$ as an $(\epsilon,k)$-row PCP of $AS_1$.
\end{lemma}

Note that by Lemma \ref{lem:sum}, $\sum_i \tau_i^{k'}(\ha) = O(k/\epsilon)$. So if $\tilde \tau_i^{k'}(\ha)$ is a constant factor approximation to $\tau_i^{k'}(\ha)$, $t = O \left (\frac{\sqrt{nk}\log n}{\epsilon^{2.5}} \right)$. Also note that the Lemma requires both $S_1$ and $S_2$ to be sampled using the rank $k'$ ridge scores. If we sample $S_1$ using a sum of the rank-$k$ and rank-$k'$ ridge scores (appropriately scaled) Lemma \ref{leveragePCP} and Lemma \ref{thm:pcp} will hold simultaneously. 

By applying an input sparsity time low-rank approximation algorithm to $\tilde A$ (which has just $\tilde O \left (\frac{nk}{\epsilon^5} \right )$ entries) we can find a near optimal low-rank approximation of $AS_1$, and thus for $A$. However, in our final algorithm, we will take a somewhat different approach. We are able to show that using appropriate sampling probabilities, we can in fact sample $\tilde A$ which is a projection-cost preserving sketch of $AS_1$ for \emph{spectral norm} error. As we will see, recovering a near optimal spectral norm low-rank approximation to $AS_1$ suffices to recover a near optimal Frobenius norm approximation to $A$, and allows us to improve $\epsilon$ dependencies in our final runtime. 

\begin{lemma}[Spectral Norm Row PCP]\label{leveragePCPspectral} For any PSD $A \in \R^{n \times n}$, and $\epsilon < 1$ let $k' = \lceil ck/\epsilon^2 \rceil$ and $\tilde \tau_i^{k'}(\ha) \ge \tau_i^{k'}(\ha)$ be an overestimate for the $i^{th}$ rank-$k'$ ridge leverage score of $\ha$. Let $\tilde \ell_i = 4\epsilon \sqrt{\frac{n}{k}} \tau_i^{k'}(\ha)$, $p_i = \frac{\tilde \ell_i}{\sum_i \tilde \ell_i}$, and $t = \frac{c '\log n}{\epsilon^2}\cdot \sum_i \tilde \ell_i$. Construct weighted sampling matrices $S_1,S_2 \in \R^{n \times t}$, each whose $j^{th}$ column is set to $\frac{1}{\sqrt{tp_i}} e_i$ with probability $p_i$.

For sufficiently large constants $c,c'$, with high probability (i.e. probability $\ge 1 -1/n^d$ for some large constant $d$), letting $\tilde A = S_2^T A S_1$, for any orthogonal projection $P \in \mathbb{R}^{t \times t}$:
\begin{align*}
(1-\epsilon) \norm{AS_1(I-P)}_2^2 - \frac{\epsilon}{k}\norm{A-A_k}_F^2 \le \norm{\tilde A(I-P)}_2^2 \le (1+\epsilon) \norm{AS_1(I-P)}_2^2 + \frac{\epsilon}{k}\norm{A-A_k}_F^2.
\end{align*}
We refer to $\tilde A$ as an $(\epsilon,k)$-spectral PCP of $AS_1$. 
\end{lemma}
Note that if $\tilde \tau_i^{k'}(\ha)$ is a constant factor approximation to $\tau_i^{k'}(\ha)$, $t = O \left (\frac{\sqrt{nk}\log n}{\epsilon^3} \right)$. 
\iffinal
\else
We defer the proofs of Lemmas \ref{leveragePCP} and \ref{leveragePCPspectral} to Appendix \ref{additional_proofs}.
\fi

\section{Full Low-Rank Approximation Algorithm}\label{sec:full}
We are finally ready to give our main algorithm for relative error low-rank approximation of PSD matrices in $\tilde O(n \poly(k/\epsilon))$ time, \nameref{mainAlgo}. We set $k_1 \eqdef \lceil ck/\epsilon \rceil$ and estimate the both the rank-$k$ and rank-$c'k_1$ ridge leverage scores of $\ha$ using the algorithm of \cite{mm16} (Step 1). If $c,c'$ are sufficiently large, sampling by the sum of these scores (Steps 2-3) ensures that $AS_1$ is an $(\epsilon,k)$-column PCP for $A$ and simultaneously, by applying Lemmas \ref{leveragePCP} and \ref{leveragePCPspectral} that $\tilde A = S_2^T A S_1$ is a row PCP in both spectral and Frobenius norm with rank $k_1$ and error $\epsilon = 1/2$ for $AS_1$

In conjunction, these guarantees ensure that we can apply an input sparsity time algorithm to $\tilde A$ (Step 4) to find a rank-$k_1$ $Z$ satisfying $\norm{AS_1 - AZZ^T}_2^2 = O (\norm{AS_1 - (AS_1)_{k_1}}_2^2 + \frac{1}{k_1}\norm{A-A_k}_F^2) = O \left (\frac{\epsilon}{k}\norm{A-A_k}_F^2 \right )$, where the final bound holds since $k_1 = \Theta(k/\epsilon)$. It is not hard to show that, due to this strong spectral norm bound, projecting $AS_1$ to $Z$ and taking the best rank-$k$ approximation in the span gives a near optimal Frobenius norm low-rank approximation to $AS_1$ and hence $A$. 

We can still not afford to read $AS_1$ in its entirety, so we employ a number of standard leverage score sampling techniques to perform this projection approximately. In Step 5, we sample $\tilde O(k/\epsilon^2)$ columns of $AS_1$ using the leverage scores of $Z$ (its row norms since it is an orthonormal matrix) to form $AS_1S_3$. We argue that there is a good rank-$k$ approximation to $AS_1$ lying in both the column span of $AS_1S_3$ and the row span of $Z^T$. In Step 6 we find a near optimal such approximation by further sampling $\tilde O(k/\epsilon^4)$ rows $AS_1$ by the leverage scores of $AS_1S_3$ (the row norms of $V$, an orthonormal basis for its span), and computing the best rank-$k$ approximation to the sampled matrix falling in the column span of $AS_1S_3$ and the row span of $Z^T$.

Finally, in Step 7 we approximately project $A$ itself to the span of this rank-$k$ approximation by first sampling by the leverage scores of the approximation (the row norms of $Q$) and projecting.

\paragraph{Algorithm 1\label{mainAlgo}}\textbf{PSD Low-Rank Approximation}
\begin{enumerate}
\item Let $k_1 = \lceil ck/\epsilon \rceil$.
For all $i \in [1,..,n]$ compute $\tilde \tau_i^k(\ha)$ and $\tilde \tau_i^{c'k_1}(\ha)$ which are constant  factor approximations to the ridge leverage scores $\tau_i^k(\ha)$ and $\tau_i^{c'k_1}(\ha)$ respectively.
\item Set $\ell^{(1)}_i = \sqrt{\frac{n}{k}} \tilde \tau_i^k(\ha) + \sqrt{\frac{n \epsilon^4}{k_1}} \tilde \tau_i^{c'k_1}(\ha)$ and $\ell^{(2)}_i = \sqrt{\frac{n}{k_1}} \tilde \tau_i^{c'k_1}(\ha)$. Set $p_i^{(1)} = \frac{\ell_i^{(1)}}{\sum_i \ell_i^{(1)}}$  and $p_i^{(2)} = \frac{\ell_i^{(2)}}{\sum_i \ell_i^{(2)}}$.
\item Set  $t_1 = \frac{c_1 \log n}{\epsilon^2} \sum_{i} \ell_i^{(1)}$ and $t_2 = c_2 \log n \sum_{i} \ell_i^{(2)}$. Sample $S_1 \in \R^{n \times t_1}$ whose $j^{th}$ column is set to $\frac{1}{\sqrt{tp^{(1)}_i}} e_i$ with probability $p^{(1)}_i$. Sample $S_2 \in \R^{n \times t_2}$ analogously using $p^{(2)}_i$. 
\item Let $\tilde A = S_2^T A S_1$, and use an input sparsity time algorithm to compute orthonormal $Z \in \R^{t_1 \times k_1}$ satisfying the spectral guarantee $\norm{\tilde A - \tilde A ZZ^T}_2^2 \le 2 \norm{\tilde A - \tilde A_{k_1}}_2^2 + \frac{2}{k_1} \norm{\tilde A - \tilde A_{k_1}}_F^2$.
\item Let $t_3 = c_3 \left (\frac{k\log (k/\epsilon)}{\epsilon} +\frac{k}{\epsilon^2}\right )$, set $p_i^{(3)} = \frac{\norm{z_i}_2^2}{\norm{Z}_F^2}$, and sample $S_3 \in \R^{t_1 \times t_3}$ where the $j^{th}$ column is set to $\frac{1}{\sqrt{t_3 p_i^{(3)}}} e_i$ with probability $p_i^{(3)}$. Compute $V \in \mathbb{R}^{n \times t_3}$ which is an orthogonal basis for the column span of $AS_1S_3$.
\item Let $p_i^{(4)} = \frac{\norm{v_i}_2^2}{\norm{V}_F^2}$ and $t_4 = c_4 \left (\frac{t_3 \log t_3}{\epsilon^2} \right )$. Sample $S_4 \in \mathbb{R}^{n \times t_4}$ where  the $j^{th}$ column is set to $\frac{1}{\sqrt{t_4 p_i^{(4)}}} e_i$ with probability $p_i^{(4)}$. Compute $W \in \mathbb{R}^{t_3 \times t_1}$ satisfying:
\begin{align*}
W = \argmin_{W | \rank(W) = k} \norm{S_4^TAS_1S_3 W Z^T - S_4^TAS_1}_F^2.
\end{align*}
\item Compute an orthogonal basis $Q \in \R^{n \times k}$ for the column span of $AS_1S_3 W$. Let $t_5 = c_5 \left (k\log k + \frac{k}{\epsilon}\right)$, set $p_i^{(5)} = \frac{\norm{q_i}_2^2}{\norm{Q}_F^2}$, where $q_i$ is the $i^{th}$ row of $Q$.
Sample $S_5 \in \R^{n \times t_5}$ where the $j^{th}$ column is set of $\frac{1}{\sqrt{t_5 p_i^{(5)}}} e_i$ with probability $p_i^{(5)}$. Solve:
\begin{align*}
N = \argmin_{N \in \R^{n \times k}} \norm{S_5^TQ N^T - S_5^T A}_F^2.
\end{align*}
\item Return $Q,N \in \R^{n \times k}$.
\end{enumerate}

\begin{theorem}[Sublinear Time Low-Rank Approximation]\label{thm:main} Given any PSD $A \in \R^{n \times n}$, for sufficiently large constants $c, c', c_1,c_2,c_3,c_4,c_5$, for any $\epsilon < 1$, \nameref{mainAlgo} accesses $O(\frac{n \cdot k\log^2 n}{\epsilon^{2.5}} + \sqrt{n} k^{1.5} \cdot  \log^2 n \cdot \poly(1/\epsilon))$ entries of $A$, runs in $\tilde O \left (\frac{nk^{\omega-1}}{\epsilon^{2(\omega-1)}} + \sqrt{n} k^{\omega-.5} \cdot \poly(1/\epsilon)\right)$ time,
 and with probability at least $9/10$ outputs $M,N \in \R^{n \times k}$ with:
$$\norm{A-MN^T}_F^2 \le (1+\epsilon)\norm{A-A_k}_F^2.$$
\end{theorem}
For simplicity, we will show $\norm{A-MN^T}_F^2 \le (1+O(\epsilon)) \norm{A-A_k}_F^2$. By scaling up the constants $c,c',c_1,c_2,c_3,c_4,c_5$ we can then scale down $\epsilon$, achieving the claimed result. 
We first show:
\begin{lemma}\label{AS1spectral} For sufficiently large constants $c,c',c_1,c_2$, with probability $98/100$, Steps 1-4 of \nameref{mainAlgo} produce $Z \in \mathbb{R}^{t_1 \times k_1}$ satisfying:
\begin{align*}
\norm{AS_1 - AS_1ZZ^T}_2^2 \le \frac{\epsilon}{k} \norm{A - A_k}_F^2.
\end{align*}
\end{lemma}
\begin{proof}
In Step 3, $S_1$ is sampled using probabilities proportional to the scores $\ell_i^{(1)} = \sqrt{\frac{n}{k}} \tilde \tau_i^k(\ha) + \sqrt{\frac{n \epsilon^4}{k_1}} \tilde \tau_i^{c'k_1}(\ha)$. The first part of this score ensures that by Lemma \ref{thm:pcp}, with high probability $AS_1$ is an $(\epsilon,k)$-column PCP of $A$. The second half, in combination with the construction of $S_2$ ensures via Lemmas \ref{leveragePCP} and \ref{leveragePCPspectral} that if $c'$ is large enough, $\tilde A = S_2^T A S_1$ is both a $(1/2,k_1)$-row PCP of $AS_1$ for Frobenius norm error with probability $99/100$ and a $(1/2,k_1)$-spectral PCP of $AS_1$ with high probability. Note that in $\ell_i^{(1)}$ the scores corresponding to $\tilde \tau_i^{c'k_1}(\ha)$ are scaled down by  an $\epsilon^2$ factor from what is required to apply Lemmas \ref{leveragePCP} and \ref{leveragePCPspectral} with rank $k_1$ and error $1/2$. However, $t_1$ is oversampled by a $1/\epsilon^2$ factor, balancing this scaling. By a union bound, all three PCP properties hold simultaneously with probability $98/100$.

By the Frobenius error PCP property of Lemma \ref{leveragePCP}, letting $ U_{k_1}$ contain the top $k_1$ row singular vectors of $AS_1$ we have:
\begin{align}\label{tildeAFrob}
\norm{\tilde A - \tilde A_{k_1}}_F^2 &\le \norm{\tilde A - \tilde A U_{k_1} U_{k_1}^T}_F^2\nonumber\\
&\le \frac{3}{2} \norm{AS_1 - (AS_1)_{k_1}}_F^2-\Delta \le 603 \norm{A-A_k}_F^2
\end{align}
by the fact that $|\Delta| \le 600 \norm{A-A_k}_F^2$ and $\norm{AS_1 - (AS_1)_{k_1}}_F^2 \le \norm{AS_1 - (AS_1)_{k}}_F^2 \le (1+\epsilon)\norm{A-A_k}_F^2 \le 2 \norm{A-A_k}_F^2$ since $AS_1$ is an $(\epsilon,k)$-column PCP of $A$.
Further, via Lemma \ref{leveragePCPspectral}, for any rank-$k_1$ projection $P$:
\begin{align}\label{spectralFull}
\frac{1}{2} \norm{AS_1(I-P)}_2^2 - \frac{1}{2k_1}\norm{A-A_{k_1}}_F^2 \le \norm{\tilde A(I-P)}_2^2 \le \frac{3}{2} \norm{AS_1(I-P)}_2^2 + \frac{1}{2k_1}\norm{A-A_{k_1}}_F^2.
\end{align}
By this bound, we have
\begin{align}\label{tildeASpectral}
\norm{\tilde A - \tilde A_{k_1}}_2^2 \le \norm{\tilde A - \tilde A U_{k_1} U_{k_1}^T}_2^2 \le \frac{3}{2}\norm{AS_1-(AS_1)_{k_1}}_2^2 + \frac{1}{2k_1} \norm{A-A_{k_1}} \le \frac{4\epsilon}{k} \norm{A-A_k}_F^2
\end{align}
where the final bound follows because if we set $c \ge 2$, $k_1 \ge \lceil 2k/\epsilon \rceil$ so $\norm{AS_1-(AS_1)_{k_1}}_2^2 \le \frac{\epsilon}{k} \norm{AS_1-(AS_1)_k}_F^2 \le \frac{2\epsilon}{k} \norm{A-A_k}_F^2$.
With \eqref{tildeAFrob} and \eqref{tildeASpectral} in place, applying \eqref{spectralFull} again,
if we compute $Z$ satisfying the guarantee in Step 4 we have:
\begin{align*}
\norm{AS_1 - AS_1 ZZ^T}_2^2 &\le 2\norm{\tilde A - \tilde AZZ^T}_2^2 + \frac{1}{k_1} \norm{A-A_{k_1}}_F^2\\
&\le 4 \norm{\tilde A - \tilde A_{k_1}}_2^2 + \frac{4}{k_1} \norm{\tilde A - \tilde A_{k_1}}_F^2 + \frac{1}{k_1} \norm{A-A_{k_1}}_F^2\tag{By the guarantee on $Z$ in Step 4}\\
&\le \frac{(4+4\cdot 603 + 1)\epsilon}{k} \cdot \norm{A-A_k}_F^2
\end{align*}
which gives the lemma after adjusting constants on $\epsilon$ by making $c$, $c'$, $c_1$ and $c_2$ sufficiently large.
\end{proof}

Lemma \ref{AS1spectral} ensures that the rank-$k$ matrix $W$ computed in Step 6 gives a near optimal low-rank approximation of $AS_1$. Specifically we have:
\begin{lemma}\label{AS1frob}
With probability $95/100$, for sufficiently large $c,c',c_1,c_2,c_3,c_4$, Step 5 of \nameref{mainAlgo} produces $W$ such that, letting $Q \in \mathbb{R}^{n\times k}$ be an orthonormal basis for the column span of $AS_1S_3 W$:
\begin{align*}
\norm{AS_1 - QQ^TAS_1}_F^2 \le (1+\epsilon) \norm{A-A_k}_F^2.
\end{align*}
\end{lemma}
\begin{proof}
We first consider the optimization problem:
\begin{align*}
M^* = \argmin_{M | \rank(M) = k} \norm{M Z^T - AS_1}_F^2.
\end{align*}
We have:
\begin{align}\label{mstar}
\norm{M^*Z^T - AS_1}_F^2 &\le \norm{(AS_1)_k ZZ^T - AS_1}_F^2 \nonumber\\
&= \norm{[AS_1 - (AS_1)_k]  + (AS_1)_k(I-Z Z^T)}_F^2\nonumber\\
& =  \norm{AS_1 - (AS_1)_k}_F^2  + \norm{(AS_1)_k(I-Z Z^T)}_F^2\nonumber\\
&\le \norm{AS_1 - (AS_1)_k}_F^2 + k \cdot \norm{AS_1(I-Z Z^T)}_2^2\nonumber\\
&\le \norm{AS_1 - (AS_1)_k}_F^2 + \epsilon \norm{A-A_k}_F^2\nonumber\\
&\le (1+2\epsilon) \norm{A-A_k}_F^2
\end{align}
where the second to last step uses that $\norm{AS_1(I-Z Z^T)}_2^2 \le \frac{\epsilon}{k} \norm{A-A_k}_F^2$ by Lemma \ref{AS1spectral} and the last step uses that $\norm{AS_1 - (AS_1)_k}_F^2 \le (1+\epsilon) \norm{A-A_k}_F^2$ since $AS_1$ is an $(\epsilon,k)$-column PCP of $A$.

Note that since it is rank-$k$ we can write $M^* = Y^* N^*$ where $Y^* \in \mathbb{R}^{n \times k}$ and $N^* \in \mathbb{R}^{k \times k_1}$ has orthonormal rows. $Y^*$ is the solution to the unconstrained low-rank approximation problem:
\begin{align*}
Y^* = \argmin_{Y \in \mathbb{R}^{n\times k}} \norm{YN^* Z^T - AS_1}_F^2.
\end{align*}
$N^* Z^T$ has orthonormal rows, so its column norms are its leverage scores. $S_3$ is sampled using the column norms of $Z^T$, which upper bound those of $N^* Z^T$. Since $\norm{Z}_F^2 = \lceil ck/\epsilon \rceil$,  $t_3 = \Theta \left ( \norm{Z}_F^2 \log(\norm{Z}_F^2) + \norm{Z}_F^2/\epsilon \right)$. It is thus well known (see e.g. Theorem 38 of \cite{clarkson2013low}, and \cite{drineas2008relative} for earlier work) that if we set
\begin{align*}
\tilde Y = \argmin_{Y \in \mathbb{R}^{n\times k}} \norm{YN^* Z^T S_3 - AS_1 S_3}_F^2
\end{align*}
we have with probability $99/100$:
\begin{align}\label{ytilde}
\norm{\tilde YN^* Z^T- AS_1}_F^2 \le (1+\epsilon)  \norm{Y^*N^* Z^T - AS_1}_F^2.
\end{align}
$\tilde Y$ can be computed in closed form as $\tilde Y = AS_1 S_3(N^* Z^TS_3)^+$, which is in the column span of $AS_1S_3$. Thus 
\eqref{ytilde} demonstrates that there is some rank-$k$ $M$ in this span satisfying $\norm{MZ^T - AS_1}_F^2 \le (1+\epsilon) \norm{M^* Z^T - AS_1}_F^2 \le (1+4\epsilon) \norm{A-A_k}_F^2$ by \eqref{mstar}. Thus if we compute
\begin{align*}
W^* = \argmin_{W | \rank(W) = k} \norm{AS_1S_3 W Z^T - AS_1}_F^2.
\end{align*}
we will have 
\begin{align}\label{wstarbound}
\norm{AS_1 S_3 W^* Z^T - AS_1}_F^2 \le (1+4\epsilon) \norm{A-A_k}_F^2.
\end{align}
We can compute $W^*$ approximately by sampling the rows of $AS_1S_3$ by their leverage scores. These are given by the row norms of the orthogonal basis $V$ as computed in Step 5 and sampled with to obtain $S_4$ in Step 6.
For any $W$, by the Pythagorean theorem, since $V$ spans $AS_1S_3$,
\begin{align}\label{affineBreakdown}
\norm{AS_1S_3 W Z^T - AS_1}_F^2 &= \norm{AS_1S_3 W Z^T - VV^TAS_1}_F^2+ \norm{(I-VV^T)AS_1}_F^2.
\end{align}
Similarly
\begin{align}\label{affineSketchBreakdown}
\norm{S_4^TAS_1S_3 W Z^T - S_4^TAS_1}_F^2 &= \norm{S_4^TAS_1S_3 W Z^T - S_4^TVV^TAS_1}_F^2+ \norm{S_4^T(I-VV^T)AS_1}_F^2 \nonumber\\&
+ 2 \tr \left ( (ZW^T S_3^TS_1^TA^T - S_1^T A^TVV^T)S_4S_4^T(I-VV^T)AS_1 \right ).
\end{align}

For the first term, since $S_4$ is sampled via the leverage scores of $AS_1S_3$ and $t_4 = c_4 \left (\frac{t_3 \log t_3}{\epsilon^2} \right )$, $S_4$ gives a subspace embedding for the column span of $AS_1S_3$ with high probability and:
\begin{align}\label{affine1}
\norm{S_4^TAS_1S_3 W Z^T - S_4^TVV^TAS_1}_F^2 \in (1 \pm \epsilon) \norm{AS_1S_3 W Z^T - VV^TAS_1}_F^2
\end{align}
For the cross term, again since $V$ spans the columns of $AS_1S_3$ we have:
\begin{align*}
 \tr \big ( (ZW^T S_3^TS_1^TA^T - S_1^T A^TVV^T)&S_4S_4^T(I-VV^T)AS_1 \big )\\
 &=  \tr \left ( (ZW^T S_3^TS_1^TA^T - S_1^T A^T)VV^TS_4S_4^T(I-VV^T)AS_1 \right )\\
&\le \norm{AS_1S_3 W Z^T - AS_1}_F \cdot \norm{VV^T S_4S_4^T(I-VV^T)AS_1}_F\\
&\le \norm{AS_1S_3 W Z^T - AS_1}_F \cdot \epsilon \norm{(I-VV^T)AS_1}_F
\end{align*}
where the last bound follows with probability $99/100$ from a standard approximate matrix multiplication result \cite{drineas2006fast} since $S_4$ is sampled using the row norms of $V$.
$\norm{(I-VV^T)AS_1}_F \le  \norm{AS_1S_3 W^* Z^T - AS_1}_F \le \norm{AS_1S_3 W Z^T - AS_1}_F$ and so overall: 
\begin{align}\label{affine3}
\tr \big ( (ZW^T S_3^TS_1^TA^T - S_1^T A^TVV^T)&S_4S_4^T(I-VV^T)AS_1 \big ) \le \epsilon \norm{AS_1S_3 W Z^T - AS_1}_F^2.
\end{align}
Combining \eqref{affineBreakdown}, \eqref{affineSketchBreakdown}, \eqref{affine1}, and \eqref{affine3} for \emph{any} $W$,
\begin{align*}
\norm{S_4^TAS_1S_3 W Z^T - S_4^TAS_1}_F^2 \in (1\pm 3\epsilon) \norm{AS_1S_3 W Z^T - AS_1}_F^2 + \Delta
\end{align*}
where $\Delta = \norm{S_4^T(I-VV^T)AS_1}_F^2 - \norm{(I-VV^T)AS_1}_F^2$ is fixed independent of $W$. Thus, if we set:
\begin{align*}
\tilde W = \argmin_{W | \rank(W) = k} \norm{S_4^TAS_1S_3 W Z^T - S_4^TAS_1}_F^2.
\end{align*}
with probability $95/100$ union bounding over all the probabilistic events above, and applying \eqref{wstarbound}:
\begin{align*}
\norm{AS_1S_3 \tilde W Z^T - AS_1}_F^2 \le \frac{(1+3\epsilon)}{(1-3\epsilon)} \norm{AS_1S_3 W^* Z^T - AS_1}_F^2 \le \frac{(1+3\epsilon)(1+4\epsilon)}{(1-3\epsilon)} \norm{A-A_k}_F^2
\end{align*}
which gives the lemma after adjusting constants on $\epsilon$ by making $c,c',c_1,c_2,c_3,c_4$ large enough.
\end{proof}

\begin{proof}[Proof of Theorem \ref{thm:main}]
By Lemma \ref{AS1frob} and since $AS_1$ is an $(\epsilon,k)$-column PCP of $A$, we have $\norm{A-QQ^TA}_F^2 \le \frac{1+\epsilon}{1-\epsilon} \norm{A-A_k}_F^2.$ In Step 7
we sample $c_5(k\log k +k/\epsilon)$ rows of Q by their standard leverage scores (their row norms) and so have with probability 99/100:
\begin{align*}
\norm{QN^T - A}_F^2 \le (1+\epsilon) \min_{Y \in \R^{n\times k}} \norm {QY^T - A}_F^2 = (1+\epsilon) \norm{A-QQ^T A}_F^2 \le \frac{(1+\epsilon)^2}{1-\epsilon} \norm{A-A_k}_F^2.
\end{align*}
Union bounding over failure probabilities and adjusting constants on $\epsilon$ yields the final claim.

It just remains to discuss \nameref{mainAlgo}'s runtime and query complexity. 
In Step 3, forming $\tilde A = S_2^T A S_1$ requires reading $t_1 \cdot t_2 = O\left (\frac{nk \log^2 n}{\epsilon^{2.5}}\right)$ entries of $A$. This dominates the access cost of Step 1, which requires $O(n k_1 \log k_1) = O\left (\frac{nk \log (k/\epsilon)}{\epsilon} \right)$ accesses by Lemma \ref{thm:originalSampling}, forming $AS_1S_3$ in Step 5 which requires $O \left (\frac{nk\log (k/\epsilon)}{\epsilon} +\frac{nk}{\epsilon^2} \right)$ accesses, and forming $S_5^T A$ in Step 7 which requires $O \left (nk\log k + \frac{nk}{\epsilon} \right )$ accesses. Forming $S_4^T A S_1$ in Step 6 requires $t_1 \cdot t_4 = O \left (\frac{\sqrt{n} k^{1.5} \log(k/\epsilon) \log n}{\epsilon^6}  \right )$ accesses, which when $n$ is large compared to $k$ and $1/\epsilon$ will be dominated by our linear in $n$ terms.

For time complexity, Step 1 requires $O(n (k_1 \log k_1)^{\omega-1}) = \tilde O \left (\frac{nk^{\omega-1}}{\epsilon^{\omega-1}}\right)$ time. Computing $Z$ in Step 4 can be done using an input sparsity time algorithm for spectral norm error with rank $k_1$ and error parameter $\epsilon' = \Theta(1)$. 
By Theorem 27 of \cite{cohen2015dimensionality} or using input sparsity time ridge leverage score sampling \cite{cohen2015ridge} in conjunction with the spectral norm PCP result of Lemma \ref{thm:spectralpcp}, the total runtime required is $O (\nnz(\tilde A) )  + \tilde O (\frac{\sqrt{n} k^{\omega-1}}{\epsilon^{\omega-1}}  )$. 

 We can compute $V$ in Step 5 in $O(nt_3^{\omega-1}) = \tilde O \left (\frac{nk^{\omega-1}}{\epsilon^{2(\omega-1)}}\right)$. In Step 6, we can compute $W$ by first multiplying $S_4^T A S_1$ by $Z$ and then multiplying by $(S_4^T A S_1 S_3)^+$ and taking the best rank-$k$ approximation of the result. The total runtime is $\tilde O \left (t_1 \cdot k^{\omega-1} +k^\omega \right) \cdot\poly(1/\epsilon) = \tilde O(\sqrt{n} k^{\omega-.5}) \cdot \poly(1/\epsilon)$. Finally, we can compute $Q$ in Step 7 by first computing a $t_3 \times k$ span of the column space of $W$, multiplying this by $AS_1S_3$, and then taking an orthogonal basis for the result, requiring total time $\tilde O(k^\omega \cdot \poly(1/\epsilon)) + O\left (\frac{nk^{\omega-1} \log(k/\epsilon)}{\epsilon^2} + nk^{\omega-1} \right)$. The final regression problem can be solved by forming the pseudoinverse of $S_5^T Q$ in $O(k^\omega (\log k + 1/\epsilon))$ time and applying it to $S_5^T A$ in $O(nk^{\omega-1} (\log k + 1/\epsilon))$ time. So overall our runtime cost with linear dependence on $n$ is dominated by the cost of Step 5. We additionally must add the $\tilde O(\sqrt{n} k^{\omega-.5}) \cdot \poly(1/\epsilon)$ term from Step 6, which only dominates if $n$ is relatively small.
\end{proof}

In many applications it is desirable that the low-rank approximation to $A$ is also symmetric and positive semidefinite.
\iffinal
\else
We show in Appendix \ref{psd_appendix}
\fi that a modification to \nameref{mainAlgo} can satisfy this constraint also in $\tilde O(n \poly(k/\epsilon))$ time. The upshot is:
\begin{theorem}[Sublinear Time Low-Rank Approximation -- PSD Output]\label{thm:psdOutput}
There is an algorithm that given any PSD $A \in \R^{n \times n}$, accesses $\tilde O \left (\frac{nk^2}{\epsilon^2} + \frac{nk}{\epsilon^3} \right)$ entries of $A$, runs in $\tilde O \left (\frac{nk^\omega}{\epsilon^\omega} + \frac{n k^{\omega-1}}{\epsilon^{3(\omega-1)}} \right)$ time and with probability at least $9/10$ outputs $M \in \R^{n \times k}$ with:
$$\norm{A-MM^T}_F^2 \le (1+\epsilon)\norm{A-A_k}_F^2.$$
\end{theorem}
\section{Query Lower Bound}\label{sec:lower}

We now present our lower bound on the number of accesses to $A$ required to compute a near-optimal low-rank approximation, matching the query complexity of Algorithm \ref{mainAlgo} up to a $\tilde O(1/\epsilon^{1.5})$ factor. 
\begin{theorem}\label{thm:lbMain}
Assume that $k,\epsilon$ are such that $nk/\epsilon = o(n^2)$. Consider any (possibly randomized) algorithm $\mathcal{A}$ that, given any PSD $A \in \R^{n\times n}$, outputs a $(1+\epsilon)$-approximate rank-$k$ approximation to $A$ (in the Frobenius norm) with probability at least $2/3$. Then there must be some input ${A}$ of which $\mathcal{A}$ reads at least $\Omega(nk/\epsilon)$ positions, possibly adaptively, in expectation. 
\end{theorem}

\subsection{Lower Bound Approach}\label{sec:sublinearLBApproach}
We prove Theorem \ref{thm:lbMain} via Yao's minimax principle \cite{yao1977probabilistic}, proving a lower bound for randomized algorithms on worst-case inputs via a lower bound for deterministic algorithms on a hard input distribution.
Specifically, we will
draw the input $A$ from a distribution over binary matrices. $A$ has all $1$'s on its diagonal, along with $k$ randomly positioned (non-contiguous) blocks of all $1$'s, each of size $\sqrt{2\epsilon n/k} \times \sqrt{2\epsilon n/k}$. In other words, $A$ is the adjacency matrix (plus identity) of a graph with $k$ cliques of size $ \sqrt{2\epsilon n/k}$, placed on random subsets of the vertices, with all other vertices isolated.

It is easy to see that every $A$ chosen from the above distribution is PSD since applying a permutation yields a block diagonal matrix, each of whose blocks is a PSD matrix (either a single $1$ entry or a rank-$1$ all $1$'s block). Additionally, for every ${A}$ chosen from the distribution, the optimal rank-$k$ approximation to $A$ projects off each of the $k$ blocks, achieving Frobenius norm error $\norm{A-A_k}_F^2 = n - k\sqrt{2\epsilon n/k} \approx n$. To match this up to a $1+\epsilon$ factor, any near-optimal rank-$k$ approximation must at least capture a constant fraction of the Frobenius norm mass in the blocks since this mass is $k \cdot 2\epsilon n/k  = 2\epsilon n$. 

Doing so requires identifying at least a constant fraction of the blocks. However, since block positions are chosen uniformly at random, and since the diagonal entries of $A$ are identical and so convey no information about the positions, to identify a single block, any algorithm essentially must read arbitrary off-diagonal entries until it finds a $1$. There are $\approx n^2$ off-diagonal entries with just $2\epsilon n$ of them $1$'s, so identifying a first block requires $\Omega(n/\epsilon)$ queries to $A$ in expectation (over the random choice of ${A}$). Since the vast majority of vertices are isolated and not contained within a block, finding this first block does little to make finding future blocks easier. So overall, the algorithm must make $\Omega(nk/\epsilon)$ queries in expectation to find a constant fraction of the $k$ blocks and output a near-optimal low-rank approximation with good probability.

While the above intuition is the key idea behind the lower bound,
a rigorous proof requires a number of additional steps and modifications,
detailed in the remainder of this section. In Section \ref{sec:primitive} we introduce the notion of a \emph{primitive approximation} to a matrix, which we will employ in our lower bound for low-rank approximation. In Section \ref{sec:lowrankLB} we show that any deterministic low-rank approximation algorithm that succeeds with good probability on the input distribution described above can be used to give an algorithm that computes a primitive approximation with good probability on a matrix drawn from a related input distribution (Lemma \ref{lem:flowercup}). This reduction yields a lower bound for deterministic low-rank approximation algorithms (Theorem \ref{thm:almostlbMain}), which gives Theorem \ref{thm:lbMain} after an application of Yao's principle.

\subsection{Primitive Approximation}\label{sec:primitive}

We first define the notion of an \emph{$\epsilon$-primitive approximation} to a matrix and establish some basic properties of these approximations. 

\begin{definition}[$\epsilon$-primitive Approximation]\label{def:primitive}
A matrix $A' \in \R^{m\times m}$ is said to be $\epsilon$-{\it primitive} for $A \in \R^{m\times m}$ if the squared Frobenius norm of $A-A'$, restricted to its off-diagonal entries is $\le \epsilon m$. Note that $A'$ is allowed to have any rank. 
\end{definition}

We also define a distribution on matrices with a randomly placed block of all one entries that will will later use in our hard input distribution construction.

\begin{definition}[Random Block Matrix]\label{def:randomBlock} For any $m,\epsilon$ with $m/\epsilon \le m^2$,
let $\mu(m,\epsilon)$ be the distribution on ${A} \in \R^{m \times m}$ defined as follows. We choose a uniformly random subset $S$ of $[m]$ where $|S| = \sqrt{16 \epsilon m}$, where we assume for simplicity that $|S|$ is an integer. We generate a random matrix $A$ by setting for each $i \neq j \in S$, $A_{i,j} = 1$. We then set $A_{i,i} = 1$ for all $i$ and set all remaining entries of $A$ to equal $0$.
\end{definition}

Note that we associate a random subset $S$ with the sampling of a matrix $A$ according to $\mu(m,\epsilon)$.
It is clear that any $A$ in the support of $\mu(m,\epsilon)$ is PSD. This is since any ${A}$ in the support of $\mu(m,\epsilon)$, after a permutation, is composed of an $|S| \times |S|$ all ones block and an $(m -|S|) \times (m-|S|)$ identity.

If ${A}'$ is $\epsilon$-{\it primitive} for ${A}$ in the support of $\mu(m,\epsilon)$, it approximates ${A}$ up to error $\epsilon m$ on its off-diagonal entries. Let $R$ denote the set of off-diagonal entries of ${A}$ restricted to the intersection of the rows and columns indexed by $S$. The error on the entries of $R$ is at most $\epsilon m$. Further, restricted to these entries, ${A}$ is an all ones matrix. Thus, on the entries of $R$, both  ${A}$ and ${A}'$ look far from an identity  matrix. Formally, we show:

\begin{lemma}\label{lem:struct}
If $A' \in \R^{m\times m}$ is $\epsilon$-primitive for an $A \in \R^{m\times m}$ in the support of $\mu(m,\epsilon)$, then $A'$ is not $\epsilon$-primitive for $I$. 
\end{lemma}
\begin{proof}
By definition, any $\epsilon$-primitive matrix ${A}'$ for ${A}$ has squared Frobenius norm restricted to ${A}$'s off-diagonal entries of $\le \epsilon m$. Let $R$ denote the set of off-diagonal entries in the intersection of the rows and columns of ${A}$ indexed by $S$. Assuming $|S| \geq 4$, which follows from the assumption in Definition \ref{def:randomBlock} that $m/\epsilon \le m^2 $, $|R| = |S|^2 - S \ge 12\epsilon m$.

Restricted to the entries in $R$, ${A}$ is an all ones matrix. Thus, on at least $8\epsilon m$ of these entries ${A}'$ must have value at least $1/2$. Otherwise it would have greater than $|R| - 4\epsilon m \ge 4\epsilon m$ entries with value $\le 1/2$ and so squared Frobenius norm error on ${A}$'s off-diagonal entries greater than $\frac{1}{2^2} \cdot 4\epsilon m \ge \epsilon m$, contradicting the assumption that it is $\epsilon$-primitive for ${A}$.

Restricted to the entries in $R$, the identity matrix is all zero. Thus, ${A}'$ has Frobenius norm error on these entries at least $\frac{1}{2^2} \cdot 8\epsilon m \ge 2\epsilon m$. Thus ${A}'$ is not $\epsilon$-primitive for ${I}$, giving the lemma.
\end{proof}

Consider ${A}$ drawn from $\mu(m,\epsilon)$ with probability $1/2$ and set to ${I}$ with probability $1/2$. By Lemma \ref{lem:struct}, any (possibly  randomized) algorithm $\mathcal{A}$ that returns an $\epsilon$-primitive approximation to ${A}$ with good probability
 can be used to distinguish with good probability whether ${A}$ is in the support of $\mu(m,\epsilon)$ or ${A} = {I}$. This is because, by Lemma \ref{lem:struct}, the output of such an algorithm can only  be correct either for some ${A}$ in the support of $\mu(m,\epsilon)$ or for the identity.
We first define the distribution over matrices to which this result applies:
\begin{definition}\label{def:gamma1}
For any $m,\epsilon$ with $m/\epsilon \le m^2$,
let $\gamma(m,\epsilon)$ be the distribution on ${A} \in \R^{m \times m}$ defined as follows. With probability $1/2$ draw ${A}$ from $\mu(m,\epsilon)$ (Definition \ref{def:randomBlock}). Otherwise, draw ${A}$ from $\nu(m)$, which is the distribution whose support is only the $m \times m$ identity matrix.
\end{definition}
Note that, as with $\mu(m,\epsilon)$, we associate a random subset $S$ with $|S| = \sqrt{16\epsilon m}$ with the sampling of a matrix ${A}$ according to $\gamma(m,\epsilon)$. $S$ is not used in the case that $A$ is drawn from $\nu(m)$.

Formally, since $\mathcal{A}$ can distinguish whether ${A}$ is in the support of $\mu(m,\epsilon)$ or $\nu(m)$ (i.e., ${A}={I}$) with good probability we can prove that the distribution of $\mathcal{A}$'s access pattern (over randomness in the input and possible randomization in the algorithm) is significantly different when it is given input ${A}$ drawn from $\mu(m,\epsilon)$ than when it is given ${A}$ drawn from $\nu(m)$.
Recall for distributions $\alpha$ and $\beta$ supported on elements $s$ of a finite set $S$, that the total variation distance $D_{TV}(\alpha, \beta) = \sum_{s \in S} |\alpha(s) - \beta(s)|$, where
$\alpha(s)$ is the probability of $s$ in distribution $\alpha$.  
\begin{corollary}\label{cor:dtv}
Suppose that a (possibly randomized) algorithm $\mathcal{A}$,
with probability at least $7/12$ over its random coin flips and random input $A \in \R^{n\times  n}$ drawn from $\gamma(m,\epsilon)$, outputs an $\epsilon$-primitive matrix for $A$. Further, suppose that $\mathcal{A}$ reads at most $r$ positions of $A$, possibly adaptively.\footnote{That is, for any input ${A}$, in any random execution, $\mathcal{A}$ reads at most $r$ entries of ${A}$.} Let $S$ be a random variable indicating the list of positions read and their corresponding values.\footnote{$S$ is determined by the random input ${A}$ and the random choices of $\mathcal{A}$. Since $\mathcal{A}$ reads at most $r$ positions of any input, we always have $|S| \le r$.} Let $L(\mu)$ denote the distribution of $S$ conditioned on $A \sim \mu(m,\epsilon)$, and let $L(\nu)$ denote the distribution of $S$ conditioned on ${A} \sim \nu(m)$.\footnote{Here and throughout, we let $A \sim \mu(m,\epsilon)$ denote the event that, when ${A}$ is drawn from the distribution $\gamma(m,\epsilon)$ (Definition \ref{def:gamma1}), which is a mixture of the distributions $\mu(m,\epsilon)$ and $\nu(m)$, that ${A}$ is drawn from $\mu(m,\epsilon)$. ${A} \sim \nu(m)$ denotes the analogous event for $\nu(m)$.} Then
$$D_{TV}(L(\mu), L(\nu)) \geq 1/6.$$ 
\end{corollary}
\begin{proof}
By Lemma \ref{lem:struct}, if algorithm $\mathcal{A}$ succeeds then its output can be used to decide if $A \sim \mu(m,\epsilon)$ or if $A \sim \nu(m)$. The success probability of any such algorithm is well-known (see, e.g., Proposition 2.58 of \cite{by02}) to be at most $1/2 + D_{TV}(L(\mu), L(\nu))/2$. Making this quantity at least $7/12$ and solving for $D_{TV}(L(\mu), L(\nu))$ proves the corollary. 
\end{proof}

\subsection{Lower Bound for Low-Rank Approximation}\label{sec:lowrankLB}

We now give a reduction, showing that any deterministic relative error $k$-rank approximation algorithm that succeeds with good probability on the hard input distribution described in Section \ref{sec:sublinearLBApproach} can be used to compute an $\epsilon$-primitive approximation to ${A}$ that is drawn from $\gamma(n/(2k),\epsilon)$ with good probability. We first formally define this input distribution.
\begin{definition}[Hard Input Distribution -- Low-Rank Approximation]\label{def:hardInput2}
Suppose $2nk/\epsilon \le n^2$.
Let $\gamma_b$ be the distribution on ${A} \in \R^{n \times n}$ determined as follows. We draw a uniformly random subset $S$ of $[n]$ where $|S| = n/2$, where we assume for simplicity that $|S|$ is an integer. We further partition $S$ into $k$ subsets $S^1,S^2,...,S^k$ chosen uniformly at random. For all $\ell \in [k]$, $|S^\ell| = n/(2k)$, which we also assume to be an integer. 

Letting ${A}^\ell$ denote the entries of ${A}$ restricted to the intersection of the rows and columns indexed by $S^\ell$, we independently draw each ${A}^\ell$ from $\gamma(n/(2k),\epsilon)$ (Definition \ref{def:gamma1}).\footnote{By the assumption that $2nk/\epsilon \le n^2$ we have $\frac{n}{2k\epsilon}\le \left( \frac{n}{2k}\right)^2$ and so this is a valid setting of the parameters for $\gamma(m,\epsilon)$.} We then set ${A}_{i,i} = 1$ for all $i$ and set all remaining entries of $A$ to equal $0$.
\end{definition}

Our reduction from $\epsilon$-primitive approximation to low-rank approximation is as follows:
\begin{lemma}[Reduction from Primitive Approximation to Low-Rank Approximation]\label{lem:flowercup}
Suppose that $nk/\epsilon = o(n^2)$ and that $\mathcal{A}$ is a deterministic algorithm that, with probability $\ge 2/3$ on random input ${A} \in \R^{n \times n}$ drawn from $\gamma_b$, outputs a $(1+\epsilon/26)$-approximate rank-$k$ approximation to ${A}$. 
Further suppose $\mathcal{A}$ reads at most $r$ positions of ${A}$, possibly adaptively. 

Then there is a randomized algorithm $\mathcal{B}$ that, with probability $\ge 7/12$ over its random coin flips and random input ${B}$ drawn from $\gamma(n/(2k),\epsilon)$, outputs an $\epsilon$-primitive matrix for ${B}$. Further, letting $L(\mu), L(\nu)$ be as defined in Corollary \ref{cor:dtv} for $\mathcal{B}$,
\begin{align*}
D_{TV}(L(\mu), L(\nu)) \le \frac{2\epsilon r}{nk}.
\end{align*}
\end{lemma}
\begin{proof}
Consider a randomized algorithm $\mathcal{B}$ that, given  ${B}$ drawn from $\gamma(n/(2k),\epsilon)$
generates a random matrix ${A}^{n \times n}$ drawn from $\gamma_b$ as follows. Choose a uniformly random subset $S$ of $[n]$ with $|S| = n/2$. Partition $S$ into $k$ subsets $S^1,S^2,...,S^k$ chosen uniformly at random. Note that for $\ell \in [k]$ $|S^\ell| = n/(2k)$. Letting ${A}^\ell$ denote the entries of ${A}$ restricted to the intersection of the rows and columns indexed by $S^\ell$, set ${A}^1 = {B}$, and for $\ell = 2,...,k$ independently draw ${A}^\ell$ from $\gamma(n/(2k),\epsilon)$. Set ${A}_{i,i} = 1$ for all $i$ and set all remaining entries of $A$ equal to $0$. 

After generating ${A}$, $\mathcal{B}$ then applies $\mathcal{A}$ to ${A}$ to compute a rank-k approximation ${A}'$. $\mathcal{B}$ then outputs $({A}')^1$, the $n/(2k) \times n/(2k)$ submatrix of ${A}'$  corresponding to the intersection of the rows and columns indexed by $S^1$. We have the following:
\begin{claim}\label{clm:flowercup1}
With probability $\ge 7/12$ over the random choices of $\mathcal{B}$ and over the random input ${B}$ drawn from $\gamma(n/(2k),\epsilon)$, $({A}')^1$ is $\epsilon$-primitive for ${B}$.
\end{claim}
\begin{proof}
It is clear that ${A}$ generated by $\mathcal{B}$ is distributed according to $\gamma_b$ (Definition \ref{def:hardInput2}). By construction, for any $A \sim \gamma_b$ there is a rank-$k$ approximation of cost at most $n$; indeed this follows by choosing the best rank-$1$ solution for each ${A}^\ell$. Consequently if $\mathcal{A}$ succeeds on ${A}$, then its output $A'$ satisfies $\|A-A'\|_F^2 \leq n + \epsilon n/26$. 

Note that ${A}$ has ${A}_{i,i} = 1$ for all $i \in [n]$. So the squared Frobenius norm cost of any rank-$k$ approximation restricted to the diagonal is at least $n-k$. Let $c_\ell$ be the squared Frobenius norm cost of ${A}'$ restricted to the off diagonal entries in ${A}^\ell$. Note that $c_\ell$ is a random variable. Then,
\begin{eqnarray*}
n-k + \sum_{\ell=1}^k c_\ell  \leq \|A'-A\|_F^2 \leq n + \epsilon n/26.
\end{eqnarray*}
By averaging for at least a $11/12$ fraction of the blocks $i$, 
$$c_i \leq 12 + 12 \epsilon n/(26k) \leq 13 \epsilon n/(26k) < \epsilon n/(2k),$$
assuming $\epsilon n / (26k) \geq 8$, which holds if $\epsilon n /k = \omega(1)$, which follows from our assumption that $nk/\epsilon = o(n^2)$. 

It follows by symmetry of $\gamma_b$ with respect to the $k$ blocks that with probability at least $11/12$, if $\mathcal{A}$ succeeds on ${A}$, $c_1 \le \epsilon n/(2k)$. This gives that $({A}')^1$ is $\epsilon$-primitive (Definition \ref{def:primitive}) for ${A}^1 = {B}$.  This yields the claim after applying a union bound, since $\mathcal{A}$ succeeds with probability at least $2/3$.
\end{proof}
It remains to bound the total variation distance between $\mathcal{B}$'s access pattern when ${B} \sim \mu(n/(2k),\epsilon)$ and when ${B} \sim \nu(n/(2k))$. We have:
\begin{claim}\label{clm:flowercup2} For the algorithm $\mathcal{B}$ defined above
\begin{align*}
D_{TV}(L(\mu), L(\nu)) \le \frac{2\epsilon r}{nk}.
\end{align*}
\end{claim}
\begin{proof}
Let $\mathcal{W}$ denote the random variable that encompasses $\mathcal{B}$'s random choices in choosing the indices in $S^2,...,S^k$ and in setting the entries in ${A}^2,...,{A}^k$. Let $\Omega$ denote the set of all possible values of $\mathcal{W}$. Let $S$ denote the random subset with $|S| = \sqrt{n/(2k) \cdot \epsilon}$ associated with the drawing of ${B}$ from $\gamma(n/(2k),\epsilon)$ (see Definition \ref{def:gamma1}). Let $\bar{S}$ denote the set of off-diagonal entries of ${A}$ in the intersection of the rows and columns indexed by $S$.

If ${B} \sim \mu(n/(2k),\epsilon)$ then all entries of $\bar{S}$ are $1$. If ${B} \sim\nu(n/(2k))$ then ${B} = {I}$ and so all entries of $\bar{S}$ are $0$. Outside of the entries in $\bar{S}$, the entries of ${B}$ are identical in the two cases that ${B} \sim \mu(n/(2k),\epsilon)$ and ${B} \sim \nu(n/(2k))$ (in particular, they are all $0$ off the diagonal and $1$ on the diagonal). 

So, for any $w \in \Omega$, conditioned on $\mathcal{W} = w$, all entries outside $\bar S$ are fixed. Further conditioned on ${B} \sim \nu(n/(2k))$, all entries in $\bar S$ are $0$. So all entries of ${A}$ are fixed and $\mathcal{A}$ always reads the same sequence of entries $(i_{w,1}, j_{w,1}), ..., (i_{w,r}, j_{w,r})$.
Further, for any $w \in \Omega$, conditioned on $\mathcal{W} = w$, $S$ is a uniform random subset of $R = [n]\setminus (S^2\cup...\cup S^k)$. $|R| = n - (k-1) \cdot n/(2k) \ge n/2$. So, for any $\ell \in [r]$ we have:
\begin{align*}
\Pr[(i_{w,\ell}, j_{w,\ell}) \in \bar{S} | \mathcal{W} = w] \le \frac{|\bar S|}{|R|^2} \le \frac{|S|^2}{n^2/4} \le \frac{2\epsilon}{nk}.
\end{align*}
By a union bound, letting $\mathcal{E}$ be the event that $\mathcal{A}$ reads $(i_{w,1}, j_{w,1}), ..., (i_{w,r}, j_{w,r})$ and does not read any entry of $\bar S$ in its $r$ accesses to ${A}$ we have for any $w \in \Omega$:
\begin{align*}
\Pr[\mathcal{E} | \mathcal{W} = w] \ge 1- \frac{2\epsilon r}{nk}.
\end{align*}
Thus, for any $w \in \Omega$, regardless of whether ${B} \sim \nu(n/(2k))$ or ${B} \sim \mu(n/(2k),\epsilon)$, $\mathcal{A}$ has access pattern  $(i_{w,1}, j_{w,1}), ..., (i_{w,r}, j_{w,r})$ to ${A}$ with probability  $\ge 1- \frac{2\epsilon r}{nk}$. Correspondingly, $\mathcal{B}$ has a fixed access pattern to ${B}$ with probability $\ge 1- \frac{2\epsilon r}{nk}$. Thus,
\begin{align*}
D_{TV}(L(\mu), L(\nu)) \le \frac{2\epsilon r}{nk}
\end{align*}
yielding the claim.
\end{proof}
In combination Claims \ref{clm:flowercup1} and \ref{clm:flowercup2} give Lemma \ref{lem:flowercup}.
\end{proof}

We can now use Lemma  \ref{lem:flowercup} to argue that if a deterministic low-rank approximation algorithm succeeding with good probability  on a random input drawn from $\gamma_b$ accesses too few entries, then it can be used to give a primitive approximation algorithm violating Corollary \ref{cor:dtv}.
\begin{theorem}[Lower Bound for Deterministic Algorithms]\label{thm:almostlbMain}
Assume that $n,k,\epsilon$ are such that $nk/\epsilon = o(n^2)$.  Consider any deterministic algorithm $\mathcal{A}$ that, given random input ${A}$ drawn from $\gamma_b$, outputs a $(1+\epsilon)$-approximate rank-$k$ approximation to $A$ (in the Frobenius norm) with probability at least $2/3$. Further, suppose $\mathcal{A}$ reads at most $r$ positions of $A$, possibly adaptively. Then $r = \Omega(nk/\epsilon)$. 
\end{theorem}
\begin{proof}
Assume towards a contradiction that $r = o(nk/\epsilon)$. Then applying 
Lemma \ref{lem:flowercup}, $\mathcal{A}$ can be used to give a randomized algorithm $\mathcal{B}$ that with probability $\ge 7/12$ over its random coin flips and random input ${B} \in n/(2k) \times n/(2k)$ drawn from $\gamma(n/(2k),\epsilon)$, outputs an $\epsilon$-primitive matrix for ${B}$. Further, letting $L(\mu), L(\nu)$ be as defined in Corollary \ref{cor:dtv} for $\mathcal{B}$, by  Lemma \ref{lem:flowercup}, $D_{TV}(L(\mu), L(\nu)) \le \frac{2\epsilon r}{nk}.$ For $r = o(nk/\epsilon)$, $\frac{2\epsilon r}{nk}= o(1)$, contradicting Corollary \ref{cor:dtv}, and giving the theorem.
\end{proof}

Theorem \ref{thm:lbMain} follows directly from applying Yao's minimax principle (\cite{yao1977probabilistic}, Theorem 3) to Theorem \ref{thm:almostlbMain}.

\section{Spectral Norm Error  Bounds}\label{sec:spectral}

We conclude by showing how to modify \nameref{mainAlgo} to output a low-rank approximation $B$ achieving the spectral norm guarantee:
\begin{align}\label{eq:spectralGuarantee}
\norm{A-B}_2^2 \le (1+\epsilon) \norm{A-A_k}_2^2 + \frac{\epsilon}{k} \norm{A-A_k}_F^2.
\end{align}
This can be significantly stronger than the Frobenius guarantee \eqref{eqn:guarantee} when $\norm{A-A_k}_F^2$ is large, and, for example is critical in our application to ridge regression, discussed in Section \ref{sec:regression}.

It is not hard to see that since additive error in the Frobenius norm upper bounds additive error in the spectral norm (see e.g. Theorem 3.4 of \cite{guNewPaper}) that for $B$ satisfying the Frobenius norm guarantee $\norm{A-B}_F^2 \le (1+\epsilon) \norm{A-A_k}_F^2$, we immediately have the spectral bound $\norm{A -B}_2^2 \le \norm{A-A_k}_2^2 + \epsilon \norm{A-A_k}_F^2$. Thus, we can achieve 
\ref{eq:spectralGuarantee} simply by running \nameref{mainAlgo} with error parameter $\epsilon/k$. However, this approach is suboptimal. Applying Theorem \ref{thm:main}, our query complexity would be $\Theta \left ( \frac{n k^{3.5} \log^2 n}{\epsilon^{2.5}} \right )$. We improve this $k$ dependence significantly in \nameref{spectralAlgo}. Since \eqref{eq:spectralGuarantee} is often applied (see for example Section \ref{sec:regression}) with $k' = k/\epsilon$ and $\epsilon = \Theta(1)$ to give $\norm{A-B}_2^2 \le O \left (\frac{\epsilon}{k} \norm{A-A_k}_F^2 \right)$, optimizing $k$ dependence is especially important.

We first give an extension of Lemma \ref{thm:pcp} to the spectral norm case. This lemma provides the column sampling analog to Lemma \ref{leveragePCPspectral}. 

\begin{lemma}[Spectral Norm PCP]\label{thm:spectralpcp} For any $A \in \R^{n \times d}$, for $i \in \{1,\ldots,d\}$, let $\tilde \tau_i^k \ge \tau_i^k(A)$ be an overestimate for the $i^{th}$ rank-$k$ ridge leverage score. 
Let $p_i = \frac{\tilde \tau^k_i}{\sum_i \tilde \tau^k_i}$ and $t = \frac{c\log(k/\delta)}{\epsilon^2} \sum_i \tilde \tau^k_i$ for any $\epsilon < 1$ and sufficiently large constant $c$. Construct $C$ by sampling $t$ columns of $A$, each set to $\frac{1}{\sqrt{tp_i}}a_i$ with probability $p_i$.
With probability $1-\delta$, for any orthogonal projection $P \in \R^{n \times n}$,
\begin{align*}
(1-\epsilon)\|A - PA\|^2_2 - \frac{\epsilon}{k} \norm{A-A_k}_F^2 \leq \|C - PC\|^2_2\leq (1+\epsilon)\|A - PA\|^2_2 + \frac{\epsilon}{k} \norm{A-A_k}_F^2.
\end{align*}
We refer to $C$ as an $(\epsilon,k)$-spectral PCP of $A$.
\end{lemma}
\begin{proof}
This follows from Corollary \ref{cor:chernoff}. With probability $\ge 1-\delta$, sampling by the rank-$k$ ridge leverage scores gives $C$ satisfying:
\begin{align}
(1-\epsilon) CC^T - \frac{\epsilon}{k} \norm{A-A_k}_F^2 I \preceq AA^T \preceq (1+\epsilon) CC^T + \frac{\epsilon}{k} \norm{A-A_k}_F^2 I.\label{spectralNormBoundForSpectral}
\end{align}

We can write for any $M$, $\norm{M}_2^2 = \max_{x:\norm{x}_2^2 = 1} x^T M x$. When $\norm{x}_2^2 = 1$, $\norm{(I-P)x}_2^2 \le 1$ so by \eqref{spectralNormBoundForSpectral} we have for any unit norm $x$:
\begin{align*}
x^T (I-P) CC^T (I-P) x  &\le x^T (I-P) AA^T (I-P) x + \frac{\epsilon}{k} \norm{A-A_k}_F^2\\
&\le  \norm{A-PA}_2^2 + \frac{\epsilon}{k} \norm{A-A_k}_F^2
\end{align*}
which gives $\norm{C-PC}_2^2 \le \frac{1}{1-\epsilon}\norm{A-PA}_2^2 + \frac{\epsilon}{(1-\epsilon) k} \norm{A-A_k}_F^2$.
Similarly we have:
\begin{align*}
x^T (I-P) AA^T (I-P) x &\le (1+\epsilon) x^T (I-P) CC^T (I-P) x +  \frac{\epsilon}{k} \norm{A-A_k}_F^2 \\ &\le (1+\epsilon) \norm{C-PC}_2^2 + \frac{\epsilon}{k} \norm{A-A_k}_F^2
\end{align*}
which gives $\frac{1}{1+\epsilon}\norm{A-PA}_2^2 - \frac{\epsilon}{(1+\epsilon) k} \norm{A-A_k}_F^2 \le \norm{C-PC}_2^2$. The lemma follows from combining these upper and lower bounds after adjusting constant factors on $\epsilon$ (by making $c$ large enough).
\end{proof}

\subsection{Spectral Norm Low-Rank Approximation Algorithm}

We now use Lemma \ref{thm:spectralpcp}, along with its row sampling counterpart, Lemma \ref{leveragePCPspectral}, to give an algorithm (\nameref{spectralAlgo}) for computing a near optimal spectral norm low-rank approximation to $A$.

In Steps 1-3 we sample both rows and columns of $A$ via the rank $\Theta(k/\epsilon)$ ridge leverage scores of $\ha$, ensuring with high probability that $AS_1$ is an $(\epsilon,k)$-spectral PCP of $A$ and $\tilde A$ is in turn an $(\epsilon,k)$-spectral row PCP of $AS_1$. Thus, if we compute (using an input sparsity time algorithm) a span $Z$ which gives a near optimal spectral norm low-rank approximation to $\tilde A$ (Step 3), this span will also be nearly optimal for $AS_1$. Since we cannot afford to fully read $AS_1$, we approximately project it to $Z$ by further sampling its columns using $Z$'s leverage scores (Step 4). We use leverage score sampling again in Step 5 to approximately project $A$ to the span of the result. This yields our final approximation, using the fact that $AS_1$ is a spectral PCP for $A$.


\paragraph{Algorithm 2\label{spectralAlgo}}\textbf{PSD Low-Rank Approximation -- Spectral Error}
\begin{enumerate}
\item Let $k_1 = \lceil ck/\epsilon^2 \rceil$.
For all $i \in [1,..,n]$ compute $\tilde \tau_i^{k_1}(\ha)$ which is a constant factor approximation to the ridge leverage score $\tau_i^{k_1}(\ha)$.
\item Set $\ell^{(1)}_i = 4 \epsilon \sqrt{\frac{n}{k}} \tilde \tau_i^{k_1}(\ha)$. Set $p_i^{(1)} = \frac{\ell_i^{(1)}}{\sum_i \ell_i^{(1)}}$ and  $t_1 = \frac{c_1 \log n}{\epsilon^2} \sum_{i} \ell_i^{(1)}$. Sample $S_1,S_2 \in \R^{n \times t_1}$ each whose $j^{th}$ column is set to $\frac{1}{\sqrt{tp^{(1)}_i}} e_i$ with probability $p^{(1)}_i$.
\item Let $\tilde A = S_2^T A S_1$, and use an input sparsity time algorithm to compute orthonormal $Z \in \R^{t_1 \times k}$ satisfying the Frobenius guarantee $\norm{\tilde A-\tilde AZZ^T}_F^2 \le 2\norm{\tilde A - \tilde A_k}_F^2$ along with the spectral guarantee $\norm{\tilde A - \tilde A ZZ^T}_2^2 \le (1+\epsilon) \norm{\tilde A - \tilde A_{k}}_2^2 + \frac{\epsilon}{k} \norm{\tilde A - \tilde A_{k}}_F^2$.
\item Let $t_3 = c_3 \left (k\log k +\frac{k^2}{\epsilon}\right )$, set $p_i^{(3)} = \frac{\norm{z_i}_2^2}{\norm{Z}_F^2}$, and sample $S_3 \in \R^{t_1 \times t_3}$ where the $j^{th}$ column is set to $\frac{1}{\sqrt{t_3 p_i^{(3)}}} e_i$ with probability $p_i^{(3)}$. Solve:
\begin{align*}
M = \argmin_{M \in \mathbb{R}^{n\times k}} \norm{AS_1S_3 - MZ^T S_3}_F^2
\end{align*}
\item Compute an orthogonal basis $Q \in \R^{n \times k}$ for the column span of $M$. Let $t_4 = c_4 \left (k\log k + \frac{k^2}{\epsilon}\right)$, set $p_i^{(4)} = \frac{\norm{q_i}_2^2}{\norm{Q}_F^2}$, and sample $S_4 \in \R^{n \times t_4}$ where the $j^{th}$ column is set to $\frac{1}{\sqrt{t_4 p_i^{(4)}}} e_i$ with probability $p_i^{(4)}$. Solve:
\begin{align*}
N = \argmin_{N \in \R^{n \times k}} \norm{S_4^TQ N^T - S_4^T A}_F^2.
\end{align*}
\item Return $Q,N \in \R^{n \times k}$.
\end{enumerate}

\begin{theorem}[Sublinear Time Low-Rank Approximation --Spectral Norm Error]\label{thm:spectral} Given any PSD $A \in \R^{n \times n}$, for sufficiently large constants $c,c_1,c_2,c_3,c_4$, \nameref{spectralAlgo} accesses $O(\frac{n \cdot k\log^2 n}{\epsilon^6} + \frac{nk^2}{\epsilon})$ entries of $A$, runs in $\tilde O \left (\frac{nk^\omega}{\epsilon} + \frac{nk}{\epsilon^6} + (\sqrt{n} k^{\omega-1} + k^{\omega+1}) \cdot \poly(1/\epsilon) \right)$ time and with probability at least $9/10$ outputs $M,N \in \R^{n \times k}$ with:
$$\norm{A-MN^T}_2^2 \le (1+\epsilon)\norm{A-A_k}_2^2 + \frac{\epsilon}{k} \norm{A-A_k}_F^2.$$
\end{theorem}
\begin{proof}
$\ell_i^{(1)} = 4\epsilon \sqrt{ \frac{n}{k}} \tilde \tau_i^{k_1}(\ha)$ which by Lemma \ref{lem:scoreBound} is within a constant factor of upper bounding the rank $k_1 = \lceil ck/\epsilon^2 \rceil$ ridge leverage scores of $A$. As long as $c \ge 1$, these scores upper bound the rank-$k$ ridge leverage scores. So by Lemma \ref{thm:spectralpcp}, with high probability $AS_1$ an $(\epsilon,k)$-spectral PCP of $A$ as long as $c_1$ is set large enough. Additionally, by Lemma \ref{leveragePCPspectral}, with high probability $\tilde A$ is an $(\epsilon,k)$-spectral row PCP of $AS_1$. 

Frobenius norm PCP bounds also hold. By Lemma \ref{thm:pcp}, $AS_1$ is an $(\epsilon,k_1)$-column PCP of $A$ with high probability and $\tilde A$ is an $(\epsilon, k_1)$-row PCP for $AS_1$ with probability $99/100$. These bounds trivially give that $AS_1$ and $\tilde A$ are $(\epsilon, k)$ PCPs for $A$ and $AS_1$ respectively, which gives:
\begin{align}\label{looseFrob}
\norm{\tilde A - \tilde A_k}_F^2 \le c_5 \norm{A-A_k}_F^2
\end{align}
for some constant $c_5$. 
By \eqref{looseFrob} and the fact that $\tilde A$ is an $(\epsilon,k)$-spectral PCP of $AS_1$, for $Z$ computed in Step 3 of \nameref{spectralAlgo} we have:
\begin{align}\label{Zbound}
\norm{AS_1 - AS_1 ZZ^T}_2^2 &\le \frac{(1+\epsilon)}{(1-\epsilon)} \left ( \norm{\tilde A - \tilde A_k}_2^2 + \frac{\epsilon}{k} \norm{\tilde A - \tilde A_k}_F^2 \right ) + \frac{\epsilon}{k(1-\epsilon)} \norm{A-A_k}_F^2\nonumber\\
&\le (1+3\epsilon) \norm{\tilde A - \tilde A_k}_2^2 + \frac{(3c_5 + 2) \epsilon}{k} \norm{A-A_k}_F^2\nonumber\\
&\le (1+3\epsilon)(1+\epsilon) \norm{AS_1 - (AS_1)_k}_2^2 + \frac{(3c_5 + 2 + 1) \epsilon}{k} \norm{A-A_k}_F^2
\end{align}
where we assume without loss of generality that $\epsilon < 1/2$ as this can be achieved by scaling our constants sufficiently.

Now, by standard approximate regression, since $S_3$ is sampled by the leverage scores of $Z$, and since $t_3 = \Theta \left ( k\log k + k/\epsilon' \right )$ for $\epsilon' = \epsilon/k$, for $M$ computed in Step 4, with probability $99/100$:
\begin{align}\label{cam12}
\norm{AS_1 - MZ^T}_F^2 \le \left (1+\frac{\epsilon}{k} \right ) \norm{AS_1-AS_1ZZ^T}_F^2.
\end{align}
This Frobenius norm bound also implies a spectral norm bound. Specifically, we can write:
\begin{align*}
\norm{AS_1 - MZ^T}_F^2  = \norm{AS_1ZZ^T - MZ^T}_F^2 +  \norm{AS_1(I-ZZ^T)}_F^2
\end{align*}
so by \eqref{cam12} we must have $\norm{AS_1ZZ^T - MZ^T}_F^2 \le \frac{\epsilon}{k} \norm{AS_1(I-ZZ^T)}_F^2$. This gives:
\begin{align*}
\norm{AS_1 - MZ^T}_2^2 &\le \norm{AS_1ZZ^T - MZ^T}_2^2 + \norm{AS_1(I-ZZ^T)}_2^2\\
&\le \norm{AS_1(I-ZZ^T)}_2^2 + \norm{AS_1ZZ^T - MZ^T}_F^2\\
&\le \norm{AS_1(I-ZZ^T)}_2^2 + \frac{\epsilon}{k} \norm{AS_1(I-ZZ^T)}_F^2.
\end{align*}
Note that $ \norm{AS_1(I-ZZ^T)}_F^2 = O(\norm{A - A_k}_F^2)$ by the Frobenius norm guarantee required in Step 3, and the fact that $\tilde A$ and $AS_1$ are $(\epsilon,k)$ PCPs for $AS_1$ and $A$ respectively.
So combining with \eqref{looseFrob} the above implies that for $Q$ spanning the columns of $M$, 
$$\norm{AS_1 - QQ^T AS_1}_2^2 \le \norm{AS_1 - MZ^T}_2^2 \le (1+O(\epsilon)) \norm{AS_1 - (AS_1)_k}_2^2 + O\left (\frac{\epsilon}{k} \right ) \norm{A-A_k}_F^2.$$
Since $AS_1$ is an $(\epsilon,k)$-spectral PCP for $A$ we also have $\norm{A - QQ^T A}_2^2 \le (1+O(\epsilon)) \norm{AS_1 - (AS_1)_k}_2^2 + O\left (\frac{\epsilon}{k} \right ) \norm{A-A_k}_F^2.$
The theorem follows by applying an identical approximate regression argument for $N$ computed in Step 5 to show that
$$\norm{A - QN^T}_2^2 \le (1+\epsilon) \norm{A - QQ^T A}_2^2 + \frac{\epsilon}{k} \norm{A-QQ^TA}_F^2 = (1+O(\epsilon)) \norm{A - A_k}_2^2 + O\left (\frac{\epsilon}{k} \right ) \norm{A-A_k}_F^2$$
with probability $99/100$.
Adjusting constants on $\epsilon$ and union bounding over failure probabilities yields the final bound.

It just remains to discuss runtime and sample complexity. Constructing $S_2^T A S_1$ in Step 3 requires reading $t_1^2 = O \left (\frac{nk\log^2 n}{\epsilon^6} \right )$ entries of $A$. Constructing $AS_1S_3$ in Step 4 and $AS_4$ in Step 5 both require reading $O \left (n k \log k + \frac{nk^2}{\epsilon} \right )$ entries. 

For runtime, computing $Z$ in Step 3 requires $O (\nnz(\tilde A)) + \tilde O ( \sqrt{n} k^{\omega-1} \cdot \poly(1/\epsilon) )$ time using an input sparsity time algorithm (e.g. by Theorem 27 of \cite{cohen2015dimensionality} or using input sparsity time ridge leverage score sampling \cite{cohen2015ridge} in conjunction with the spectral norm PCP result of Lemma \ref{thm:spectralpcp}). Computing $M$ in Step 4 and $N$ in Step 5 both require computing the pseudoinverse of a $k \times O \left(k\log k + \frac{k^2}{\epsilon}\right)$ matrix in $\tilde O(k^{\omega+1} \poly(1/\epsilon))$ time, and then applying this to an $n \times O \left(k\log k + \frac{k^2}{\epsilon}\right)$ matrix requiring $\tilde O \left (\frac{nk^{\omega}}{\epsilon} \right )$ time.
\end{proof}

\subsection{Sublinear Time Ridge Regression}\label{sec:regression}

We now demonstrate how Theorem \ref{thm:spectral} can be leveraged to give a sublinear time, relative error algorithm for approximately solving the ridge regression problem:
\begin{align}\label{def:ridge}
\min_{x \in \mathbb{R}^n} \norm{Ax - y}_2^2 + \lambda\norm{x}_2^2
\end{align}
for PSD $A$. We begin with a lemma showing that any approximation to $A$ with small spectral norm error can be used to approximately solve \eqref{def:ridge} up to relative error.

\begin{lemma}[Ridge Regression via Spectral Norm Low-Rank Approximation]\label{ridgeLem}
For any PSD $A \in \mathbb{R}^{n \times n}$, $y \in \mathbb{R}^n$, regularization parameter $\lambda \ge 0$ and $B$ with $\norm{A-B}_2^2 \le \epsilon^2 \lambda$, let $\tilde x \in \mathbb{R}^n$ be any vector satisfying:
$$\norm{B\tilde x-y}_2^2 + \lambda \norm{\tilde x}_2^2 \le (1+\alpha) \cdot \min_{x \in \mathbb{R}^n} \norm{Bx - y}_2^2 + \lambda \norm{x}_2^2$$
we have:
\begin{align*}
\norm{A\tilde x -y}_2^2 + \lambda\norm{\tilde x}_2^2 \le (1+\alpha)(1+5\epsilon) \min_{x \in \mathbb{R}^n} \norm{Ax - y}_2^2 + \lambda \norm{x}_2^2.
\end{align*}
\end{lemma}
\begin{proof}
For any $x \in \mathbb{R}^n$ we have:
\begin{align*}
\norm{Bx - y}_2^2 + \lambda \norm{x}_2^2 = \norm{Ax - y}_2^2 + \norm{(B-A)x }_2^2 + 2x^T(B-A)^T (Ax-y) + \lambda \norm{x}_2^2
\end{align*}
Now $\norm{(B-A)x }_2^2 \le \epsilon^2 \lambda \norm{x}_2^2$ and further
\begin{align*}
|2x^T(B-A)^T (Ax-y)| &\le 2 \norm{(A-B)x}_2 \norm{Ax-y}_2\\
&\le 2\epsilon \sqrt{\lambda} \norm{x}_2 \norm{Ax-y}_2\\
&\le \epsilon \left ( \norm{Ax-y}_2^2 + \lambda \norm{x}_2^2 \right ). 
\end{align*}
So for any $x$, $\norm{Bx - y}_2^2 + \lambda \norm{x}_2^2 \in (1 \pm 2\epsilon) \left (\norm{Ax - y}_2^2 + \lambda \norm{x}_2^2  \right )$ which gives the lemma since any nearly optimal $\tilde x$ for the ridge regression problem on $B$ will also be nearly  optimal for $A$.
\end{proof}

Combining Lemma \ref{ridgeLem} with Theorem \ref{thm:spectral} gives the following:
\begin{theorem}[Sublinear Time Ridge Regression]\label{thm:regression} Given any PSD $A \in \mathbb{R}^{n\times n}$, regularization parameter $\lambda \ge 0$, $y \in \mathbb{R}^n$, and upper bound $\tilde s_\lambda$ on the statistical dimension  $s_\lambda \eqdef \tr ((A^2+\lambda I)^{-1} A^2)$, there is an algorithm accessing $\tilde O \left (\frac{n \tilde s_\lambda^2}{\epsilon^4} \right)$ entries of $A$ and running in  $\tilde O \left (\frac{n \tilde s_\lambda^\omega}{\epsilon^{2\omega}}\right)$ time, which outputs $\tilde x$ satisfying:
\begin{align*}
\norm{A\tilde x-y}_2^2 + \lambda \norm{\tilde x}_2^2 \le (1+\epsilon) \cdot \min_{x\in\mathbb{R}^n} \norm{A x-y}_2^2 + \lambda \norm{x}_2^2.
\end{align*}
\end{theorem}
When $\tilde s_\lambda \ll n$ as is often the case, the above significantly improves upon state-of-the-art input sparsity time runtimes for general matrices \cite{HCW}.
\begin{proof}
Let $k = \frac{c \tilde s_\lambda}{\epsilon^2}$ for sufficiently large constant $c$. We have:
\begin{align*}
s_\lambda = \sum_{i=1}^n \frac{\lambda_i^2(A)}{\lambda_i^2(A) + \lambda} \ge \sum_{i : \lambda^2_i(A) \ge \epsilon^2 \lambda} \frac{\lambda_i^2(A)}{(1+1/\epsilon^2)\lambda_i^2(A)} \ge \frac{\epsilon^2}{2} \cdot | \{ i :  \lambda_i^2(A) \ge\epsilon^2\lambda \} |.
\end{align*}
So $| \{ i :  \lambda_i^2(A) \ge \epsilon^2\lambda \} | \le \frac{2s_\lambda}{\epsilon^2}$ and for large enough $c$ and $k = \frac{c \tilde s_\lambda}{\epsilon^2} \ge \frac{c s_\lambda}{\epsilon^2}$ we have $\norm{A-A_k}_2^2 \le \frac{\epsilon^2 \lambda}{2} $ and can run \nameref{spectralAlgo} with error parameter $\epsilon = \Theta(1)$ to find $M,N \in \mathbb{R}^{n \times k}$ with $\norm{A-MN^T}_2^2 \le \epsilon^2\lambda$. We can then apply Lemma \ref{ridgeLem} -- solving $\tilde x = \min_{x\in\mathbb{R}^n} \norm{MN^T x-y}_2^2 + \lambda \norm{x}_2^2$ exactly using an SVD in $O(n k^{\omega-1})$ time. $\tilde x$ will be a $(1+O(\epsilon))$ approximate solution for $A$, which, after adjusting constants on $\epsilon$, gives the lemma. The runtime follows from Theorem \ref{thm:spectral} with $k = c\tilde s_\lambda/\epsilon^2$ and $\epsilon' = \Theta(1)$. The $O(n k^{\omega-1})$ regression cost is dominated by the cost of computing the low-rank approximation.
\end{proof}

Note that Theorem \ref{thm:regression} ensures that if the $k \ge \frac{c s_\lambda}{\epsilon^2}$ for some constant $c$, $\tilde x$ is a good approximation to the ridge regression problem. Setting $k$ properly requires some knowledge of an upper bound $\tilde s_\lambda$ on $s_\lambda$. A constant factor approximation to $s_\lambda$ can be computed in $\tilde O(n^{3/2} \cdot \poly(s_\lambda))$ time using for example a column PCP as given by Lemma \ref{thm:pcp} and binary  searching for an appropriate $k$ value.

An interesting open question is if $s_\lambda$ be be approximated more quickly -- specifically with linear dependence on $n$. This question is closely related to if it is possible to estimate the cost $\norm{A-A_k}_F^2$ in $\tilde O(n \cdot \poly(k))$, which surprisingly is also open.

\section*{Acknowledgements}
The authors thank IBM Almaden where part of this work was done. 
David Woodruff also thanks the Simons Institute program on Machine Learning and 
the XDATA program of DARPA for support.

\bibliographystyle{alpha}
\bibliography{sublinearLowRank}
\clearpage

\appendix
\section{Low-Rank Approximation of $A$ via Approximation of $A^{1/2}$}\label{sec:example}
We first observe that a low-rank approximation for $A^{1/2}$ 
does not imply a good low-rank approximation for $A$. Intuitively, if $A$ has a large top singular value, the low-rank approximation for $A$ must capture the corresponding singular direction with significantly more accuracy than a good low-rank approximation for $A^{1/2}$, in which the singular value is relatively much smaller.
\begin{theorem}
For any
$k$, $\epsilon$ there exists a PSD matrix $A$ and a rank $k$ matrix $B$ such that $\norm{A^{1/2}-B}_F^2 \le (1+\epsilon) \norm{A^{1/2}-A_k^{1/2}}_F^2$ but for every matrix $C$ in the rowspan of $B$, 
$$\norm{A-C}_F^2 \ge \left (1+\epsilon \cdot  \frac{(n-k-1) \lambda_1(A)}{\lambda_{k+1}(A)} \right ) \norm{A-A_k}_F^2.$$
\end{theorem}
Notably, if we set $B = \ha P$ for some rank $k$ orthogonal projection $P$, 
$AP$ can be an arbitrarily bad low-rank approximation of $A$.
\begin{proof}
Let $A \in \mathbb{R}^{n \times n}$ be a diagonal matrix with $A_{i,i} =\alpha^2$ for $i=1,...,k$, $A_{k+1,k+1} = 0$ and all other diagonal entries equal to $\beta^2$, where $\alpha > \beta > 0$. Let $B$ be a rank $k$ matrix which has its last $n-k$ rows all zero. For $i = 1,...,k$, let $B_{i,i} = \ha_{i,i}$ and $B_{1,k+2} = \sqrt{\epsilon(n-k-1)} \cdot \beta$. We have: $\norm{A^{1/2}-A^{1/2}_k}_F^2 = (n-k-1)\beta^2$ and $\norm{A^{1/2}-B}_F^2 = (1+\epsilon) (n-k-1)\beta^2$. Note that the first row $b_1$ aligns somewhat well with $a_1$, but as we will see, not well enough to give a good low-rank approximation for $A$ itself.

Let $C$ be the projection of $A$ onto the rowspan of $B$, which gives the optimal low-rank approximation to $A$ within this span. For $i=2,...,k$, $c_i = a_i$, since $A$ and $B$ match exactly on these rows up to a scaling. For $i > k$, $c_i = \vec{0}$. Finally, $c_1$ = $\frac{b_1}{\norm{b_1}_2^2} \cdot \langle b_1, a_1 \rangle = b_1\cdot \left ( \frac{\alpha^3}{\alpha^2 + \epsilon(n-k-1) \beta^2} \right )$. Overall: 
\begin{align*}
\norm{A-C}_F^2 &= (n-k-1)\beta^4 + (A_{1,1} - C_{1,1})^2 + (A_{1,k+2} - C_{1,k+2})^2\\ 
&\ge (n-k-1)\beta^4 + \left (\frac{\sqrt{\epsilon(n-k-1)}\cdot\beta \alpha^3}{\alpha^2 + \epsilon(n-k-1) \beta^2} \right )^2\\ 
&\ge (n-k-1)\beta^4 \cdot (1 + \epsilon (n-k-1)\alpha^2/4\beta^2)\\
&= (1 + \epsilon (n-k-1)\alpha^2/\beta^2) \cdot \norm{A-A_k}_F^2.
\end{align*}
By setting $\alpha \gg \beta$ we can make this approximation arbitrarily bad. Note that $\alpha^2/\beta^2 = \lambda_1(A)/\lambda_{k+1}(A)$. This ratio will be large whenever $A$ is well approximated by a low-rank matrix.
\end{proof}

Despite the above example, we can show that for a projection $P$, if $\ha P$ is a very near optimal low-rank approximation of $\ha$ then  $\ha P \ha$ is a relative error low-rank approximation of $A$:

\begin{theorem}\label{APAbound}
Let $P \in \R^{n \times n}$ be an orthogonal projection matrix such that $\norm{A^{1/2}-A^{1/2}P}_F^2 \le (1+\epsilon/\sqrt{n}) \norm{A^{1/2}-A^{1/2}_k}_F^2$. Then:
\begin{align*}
\norm{A-A^{1/2}PA^{1/2}}_F^2 \le (1+3\epsilon) \norm{A-A_k}_F^2.
\end{align*}
\end{theorem}
\begin{proof}
We can rewrite using the fact that $(I-P) = (I-P) ^2$ since it is a projection:
\begin{align*}
\norm{A-A^{1/2}PA^{1/2}}_F^2 = \norm{A^{1/2}(I-P)^2 A^{1/2}}_F^2 = \norm{A^{1/2} (I-P)}_4^4.
\end{align*}
Let $\delta_i = \sigma_i(A^{1/2}(I-P))$ denote the $i^{th}$ singular value of  $A^{1/2}(I-P)$. Let $\lambda_i$ be the $i^{th}$ eigenvalue of $A$.
By the assumption that $P$ gives a near optimal low-rank approximation of $\ha$:
\begin{align*}
\sum_{i=1}^{n-k} \delta^2_i \le \sum_{i=k+1}^{n} \lambda_i + \epsilon/\sqrt{n} \norm{A^{1/2}-A^{1/2}_k}_F^2.
\end{align*}
Additionally, by Weyl's monotonicity theorem (see e.g. Theorem 3.2 in \cite{guNewPaper} and proof of Lemma 15 in \cite{musco2015randomized}), for all $i$, $\delta_i \ge  \lambda^{1/2}_{i+k}$. We thus have:
\begin{align*}
\norm{A-A^{1/2}PA^{1/2}}_F^2 = \sum_{i=1}^{n-k} \delta_i^4 \le \sum_{i=k+2}^{n} \lambda_i^2 + \left (\lambda_{k+1} + \epsilon/\sqrt{n} \norm{A^{1/2}-A^{1/2}_k}_F^2 \right )^2.
\end{align*}
If $\lambda_{k+1} \ge 1/\sqrt{n} \cdot \norm{A^{1/2}-A^{1/2}_k}_F^2$ then $\left (\lambda_{k+1} + \epsilon/\sqrt{n} \norm{A^{1/2}-A^{1/2}_k}_F^2 \right )^2 \le (1+ \epsilon)^2 \lambda_{k+1}^2 \le (1+3\epsilon) \lambda_{k+1}^2$ and hence:
\begin{align*}
\norm{A-A^{1/2}PA^{1/2}}_F^2 \le (1+3\epsilon)\sum_{i=k+1}^n \lambda_i^2 = (1+3\epsilon) \norm{A-A_k}_F^2.
\end{align*}
Alternatively if $\lambda_{k+1} \le 1/\sqrt{n} \cdot \norm{A^{1/2}-A^{1/2}_k}_F^2$ then $$\left (\lambda_{k+1} + \epsilon/\sqrt{n} \norm{A^{1/2}-A^{1/2}_k}_F^2 \right )^2 \le \left ((1+\epsilon)/\sqrt{n} \norm{A^{1/2}-A^{1/2}_k}_F^2 \right )^2 \le (1+\epsilon)^2 \norm{A-A_k}_F^2$$ which also gives the theorem.
\end{proof}
\subsection{PSD Low-Rank Approximation in $n^{1.69} \cdot \poly(k/\epsilon)$ Time}

We now combine Theorem \ref{APAbound} with the ridge leverage score sampling algorithm of \cite{mm16} to give a sublinear time algorithm for low-rank approximation of $A$ reading $n^{3/2}\cdot \poly(k/\epsilon)$ entries of the matrix and running in $n^{1.69} \cdot \poly(k/\epsilon)$ time. We note that this approach could also be used with adaptive sampling \cite{deshpande2006adaptive} or volume sampling techniques \cite{agr16}, as outlined in the introduction.
\begin{theorem}\label{slowRuntimeRidge}
There is an algorithm based off ridge leverage score sampling which, given PSD $A \in \R^{n \times n}$ with high probability outputs $M \in \mathbb{R}^{n \times k}$ with $\norm{A-MM^T}_F^2 \le (1+\epsilon) \norm{A-A_k}_F^2$. The algorithm reads $\tilde O(n^{3/2} k/\epsilon)$ entries of $A$ and runs in $\tilde O \left (n^{(\omega+1)/2} \cdot (k/\epsilon)^{\omega-1} \right )$ time where $\omega < 2.38$ is the exponent of matrix multiplication.
\end{theorem}
This follows from Lemma \ref{thm:originalSampling}, adapted from Theorem 8 of \cite{mm16}, which shows that it is possible to estimate the ridge leverage scores of $\ha$ with $\tilde O(nk)$ accesses to $A$ and $O(nk^{\omega-1})$ time. We can use these scores to sample a set of rows from $\ha$ whose span contains a near optimal low-rank approximation. Specifically we have:

\begin{lemma}[Theorem 7 of \cite{cohen2015ridge}]\label{css_intro} 
For any $B \in \mathbb{R}^{n \times n}$,
for $i \in \{1,\ldots,n\}$, let $\tilde \tau_i^k \ge \tau_i^k(B)$ be an overestimate for the $i^{th}$ rank-$k$ ridge leverage score of $B$. Let $p_i = \frac{\tilde \tau^k_i}{\sum_i \tilde \tau^k_i}$ and $t = c \left(\log k + \frac{\log(1/\delta)}{\epsilon}\right) \sum_i^k \tilde \tau^k_i$ for $\epsilon < 1$ and some sufficiently large constant $c$. Construct $R$ by sampling $t$ rows of $B$, each set to row $b_i$ with probability $p_i$. With probability $1-\delta$, letting $P_R$ be the projection onto the rows of $R$,
\begin{align*}
\norm{B-\left (B P_R\right)_k}_F^2 \le (1+\epsilon) \norm{B-B_k}_F^2.
\end{align*}
\end{lemma}

Note that $(BP_R)_k$ can be written as a row projection of $B$ -- onto the top $k$ singular vectors of $BP_R$.
So, if we compute for each $i$, $\tilde \tau_i^k \ge \tau_i^k(\ha)$ using Lemma \ref{thm:originalSampling}, set $\epsilon' = \epsilon/3\sqrt{n}$, and let $S$ be a sampling matrix selecting  $\tilde O( \sum \tilde \tau^k_i /\epsilon') = \tilde O(k\sqrt{n}/\epsilon)$ rows of $\ha$, then by Theorem \ref{APAbound}, letting $P$ be the projection onto the rows of $S\ha $, we have $\norm{A- (\ha P)_k (\ha P)_k^T}_F^2 \le (1+\epsilon) \norm{A-A_k}_F^2$.

Finally, $(\ha P)_k (\ha P)_k^T = (\ha P \ha)_k = (AS (S^TAS)^+ S^T A)_k$. We can compute a factorization of this matrix by computing $(S^T A S)^{+/2}$, then computing $AS (S^T A S)^{+/2}$ and taking the SVD of this matrix. Since $S$ has $\tilde O(k\sqrt{n}/\epsilon)$ columns, using fast matrix multiplication this requires time $\tilde O( n \cdot (k\sqrt{n}/\epsilon)^{\omega-1}) = \tilde O(n^{(\omega+1)/2} \cdot (k/\epsilon)^{\omega-1})$ and $\tilde O(n^{3/2}k/\epsilon)$ accesses to $A$ (to read the entries of $AS$), giving Theorem \ref{slowRuntimeRidge}.
%

\section{Additional Proofs for Main Algorithm}\label{additional_proofs}

\begin{replemma}{lem:sum}[Sum of Ridge Leverage Scores]
For any $A \in \R^{n\times d}$, $\sum_{i=1}^d \tau_i(A) \le 2k$.
\end{replemma}
\begin{proof}
We rewrite Definition \ref{def:ridgeScores}
using $A$'s singular value decomposition $A = U\Sigma V^T$.
\begin{align*}
\tau_i(A) &= a_i^T \left (U \Sigma^2 U^T + \frac{\norm{A-A_k}_F^2}{k} U U^T \right)^{-1} a_i\\
&= a_i^T \left (U \bar \Sigma U^T \right ) a_i,
\end{align*}
where $\bar \Sigma_{i,i} = \frac{1}{\sigma^2_i(A) + \frac{\norm{A-A_k}_F^2}{k}}$.
We then have:
\begin{align*}
\sum_{i=1}^n \tau_i(A) = \tr\left (A^T U \bar \Sigma U^T A\right ) = \tr \left (V \Sigma \bar \Sigma \Sigma V^T \right) = \tr(\Sigma^2 \bar \Sigma)
\end{align*}
$(\Sigma^2 \bar \Sigma)_{i,i} = \frac{\sigma_i^2(A)}{ \sigma^2_i(A) + \frac{\norm{A-A_k}_F^2}{k}}$. For $i \le k$ we simply upper bound this by $1$. So:
\begin{align*}
\tr(\Sigma^2 \bar \Sigma) = k + \sum_{i=k+1}^n\frac{\sigma_i^2(A)}{\sigma_i^2(A) + \frac{\norm{A-A_k}_F^2}{k}} \leq k +k \sum_{i=k+1}^n\frac{\sigma_i^2(A)}{\norm{A-A_k}_F^2} = 2k.
\end{align*}
\end{proof}

We first prove our row sampling PCP result for spectral norm error. We follow with the closely related proof Lemma \ref{leveragePCP} which gives Frobenius norm error.
\begin{replemma}{leveragePCPspectral}[Spectral Norm Row PCP] For any PSD $A \in \R^{n \times n}$, and $\epsilon < 1$ let $k' = \lceil c k/\epsilon^2 \rceil$ and $\tilde \tau_i^{k'}(\ha) \ge \tau_i^{k'}(\ha)$ be an overestimate for the $i^{th}$ rank-$k'$ ridge leverage score of $\ha$. Let $\tilde \ell_i = 4\epsilon \sqrt{\frac{n}{k}} \tau_i^{k'}(\ha)$, $p_i = \frac{\tilde \ell_i}{\sum_i \tilde \ell_i}$, and $t = \frac{c' \log n}{\epsilon^2}\cdot \sum_i \tilde \ell_i$. Construct weighted sampling matrices $S_1,S_2 \in \R^{n \times t}$, where the $j^{th}$ column is set to $\frac{1}{\sqrt{tp_i}} e_i$ with probability $p_i$.

For sufficiently large constants $c,c'$, with high probability, letting $\tilde A = S_2^T A S_1$, for any orthogonal projection $P \in \mathbb{R}^{t \times t}$:
\begin{align*}
(1-\epsilon) \norm{AS_1(I-P)}_2^2 - \frac{\epsilon}{k}\norm{A-A_k}_F^2 \le \norm{\tilde A(I-P)}_2^2 \le (1+\epsilon) \norm{AS_1(I-P)}_2^2 + \frac{\epsilon}{k}\norm{A-A_k}_F^2.
\end{align*}
We refer to $\tilde A$ as an $(\epsilon,k)$-spectral PCP of $AS_1$.
\end{replemma}
Note that if $\tilde \tau_i^{k'}(\ha)$ is a constant factor approximation to $\tau_i^{k'}(\ha)$, $t = O \left (\frac{\sqrt{nk}\log n}{\epsilon^3} \right)$. We use `with high probability' to mean with probability $\ge  1-1/n^d$ for some large constant $d$.
\begin{proof}
For conciseness write $C \eqdef AS_1$ and write the eigendecomposition $A=U\Lambda U^T$ with $\lambda_i = \Lambda_{i,i}$.
Applying the triangle inequality we have:
\begin{align*}
\norm{\tilde A (I-P)}_2^2 = \norm{ (I-P) \tilde A^T \tilde A (I-P)}_2 &= \norm{ (I-P) [C^T C + (\tilde A^T \tilde A - C^T C)] (I-P)}_2\\
&\in \norm{C(I-P)}_2^2 \pm \norm{ (I-P) (\tilde A^T \tilde A - C^T C) (I-P)}_2
\end{align*}
Thus to show the Lemma it suffices to show:
\begin{align}\label{sufficientSpectralBound}
 \norm{ (I-P)(\tilde A^T \tilde A - C^T C)(I-P)}_2 \le \epsilon \norm{C(I-P)}_2^2 + \frac{\epsilon}{k} \norm{A-A_k}_F^2.
\end{align}

Let $m$ be the largest index with $\lambda_m^2 \ge \frac{\epsilon^2}{k} \norm{A-A_k}_F^2$. Let $U_{H} \eqdef U_m$ contain the top $m$ `head' eigenvectors of $A$ and let $U_T$ contain the remaining `tail' eigenvectors. Let  $C_H = U_HU_H^T C$ and $C_T = U_T U_T^T C$. $C_H + C_T = C$ and so
by the triangle inequality:
\begin{align}\label{spectralBreakdown}
 \| (I-P)(\tilde A^T \tilde A - C^T C)(I-P) \|_2 \le 
\|(I-P)&(C_H^T S_2 S_2^T C_H - C_H^T C_H)(I-P)\|_2\nonumber\\ + 
&\norm{(I-P)(C_T^T S_2 S_2^T C_T - C_T^T C_T)(I-P) }_2\nonumber \\ +
&2\norm{(I-P) C_H^T S_2 S_2^T C_T (I-P)}_2. 
\end{align}

We bound each of the terms in the above sum separately. Specifically we show:
\begin{itemize}
\item Head Term: $\|(I-P)(C_H^T S_2 S_2^T C_H - C_H^T C_H)(I-P)\|_2 \le \epsilon \norm{C (I-P)}_2^2$
\item Tail Term:  $\norm{(I-P)(C_T^T S_2 S_2^T C_T - C_T^T C_T)(I-P) }_2 \le  \frac{13\epsilon^2}k  \norm{A-A_k}_F^2$.
\item Cross Term: $\norm{(I-P) C_H^T S_2 S_2^T C_T (I-P)}_2 \le \frac{10\epsilon}{k} \norm{A-A_k}_F^2 + 4\epsilon \norm{C(I-P)}_F^2$.
\end{itemize}

Combining these three bounds, after adjusting constant factors on $\epsilon$ by making the constants $c$ and $c'$ in the rank parameter $k'$ and sample size $t$ large enough, gives \eqref{sufficientSpectralBound} and thus the lemma. For the remainder of the proof we thus fix $c = 1$ so $k' = \lceil k/\epsilon^2 \rceil$.

\medskip
\noindent
\textbf{Head Term:}
\medskip

We first show that the ridge scores of $\ha$ upper bound the standard leverage scores of $U_m$.
\begin{lemma}\label{lem:aux1}
For any $p$ with $\lambda^2_p(A) \ge \frac{1}{k} \norm{A-A_k}_F^2$ we have: 
$$\sqrt{\frac{16n}{k}} \cdot \tau_i^k(\ha) \ge \norm{(U_{p})_i}_2^2.$$
where $\norm{(U_{p})_i}_2^2$ is the $i^{th}$ row norm of $U_{p}$, whose columns are the top $p$ eigenvectors of $A$.
\end{lemma}
\begin{proof}
If $U_{p}$ contains no eigenvectors, this is true vacuously as all row norms are $0$. Otherwise
\begin{align*}
\tau_i^k(A) &= a_i^T \left (A^2 + \frac{\norm{A-A_k}_F^2}{k} I \right)^{-1} a_i\nonumber\\
&= e_i^TU \hat \Lambda U^T e_i
\end{align*}
where 
$
\hat \Lambda_{j,j} = \frac{\lambda_j^2}{\lambda_j^2 + \frac{\norm{A-A_k}_F^2}{k}}.
$ We can then write:
\begin{align*}
\tau^k_i(A)= \sum_{j=1}^n U_{i,j}^2\cdot \hat \Lambda_{j,j} &\ge \sum_{j = 1}^p U_{i,j}^2\cdot \hat \Lambda_{j,j}\tag{truncate sum}\\
&\ge \sum_{j =1}^p  \left (U_{i,j}^2\cdot \frac{\lambda_j^2}{2\lambda_j^2} \right )\tag{By assumption, $\lambda_j^2(A) \ge \lambda_p^2(A) \ge \frac{1}{k}\norm{A-A_k}_F^2$}\\
&\ge  \frac{1}{2} \sum_{j=1}^m U_{i,j}^2 = \frac{1}{2} \norm{(U_{p})_i}_2^2.
\end{align*}
This gives the lemma combined with the fact that $\tau_i^k(A) \le 2\sqrt{\frac{n}{k}} \tau_i^k(\ha)$ by Lemma \ref{lem:scoreBound}.
\end{proof}

Applying Lemma \ref{lem:aux1} with $k' = \lceil k/\epsilon^2 \rceil$, since $\lambda_m^2(A) \ge \frac{1}{k'} \norm{A-A_{k'}}_F^2$, we have:
\begin{align*}
\tilde \ell_i = \sqrt{\frac{16n\epsilon^2}{k}} \cdot  \tilde \tau_i^{k'}(\ha) \ge \norm{(U_m)_i}_2^2 =  \norm{(U_H)_i}_2^2
\end{align*}
Further, $t = \frac{c'\log n}{\epsilon^2 } \cdot \sum_i \tilde \ell_i $ so if we set $c'$ large enough, we have by a standard matrix Chernoff bound (see Lemma \ref{lem:chernoff}) since $U_H$ is an orthogonal span for the columns of $C_H$ and hence and its row norms are the leverage scores of $C_H$, with high probability:
\begin{align}\label{Csubspace}
(1-\epsilon) C_H^T C_H \preceq C_H^TS_2S_2^TC_H \prec (1+\epsilon)C_H^TC_H.
\end{align}
This in turn gives:
\begin{align}
\norm{(I-P)(C_H^T S_2 S_2^T C_H - C_H^T C_H)(I-P) }_2 &\le \epsilon \norm{(I-P)C_H^T C_H(I-P) }_2\nonumber\\
 &= \epsilon \norm{C_H (I-P)}_2^2 \le \epsilon \norm{C (I-P)}_2^2.\label{spectralHead}
\end{align}

\medskip
\noindent
\textbf{Tail Term:}
\medskip

We can loosely bound via the triangle inequality and the fact that  $\norm{I-P}_2 \le 1$:
\begin{align}
\norm{(I-P)(C_T^T S_2 S_2^T C_T - C_T^T C_T)(I-P) }_2 \le \norm{S_2C_T}_2^2 + \norm{C_T}_2^2.\label{tailLoose}
\end{align}

Since $k' = \lceil k/\epsilon^2 \rceil$:
\begin{align*}
\tilde \ell_i = 4\epsilon \sqrt{\frac{n}{k}} \cdot  \tilde \tau_i^{k'}(\ha)  \ge x_i^T \left (A + \frac{\epsilon \norm{\ha-\ha_k}_F^2}{\sqrt{nk}} \right )^+ x_i.
\end{align*}
Thus for $S = S_1$ or $S = S_2$, for sufficiently large $c'$ in our sample size $t = \frac{c'\log n}{\epsilon^2}$, by a matrix Chernoff bound (Corollary \ref{cor:chernoff}), with high probability
\begin{align*}
(1-\epsilon) \ha S S^T \ha - \frac{\epsilon \norm{\ha-\ha_k}_F^2}{\sqrt{nk}} \preceq A \preceq (1+\epsilon) \ha S S^T \ha + \frac{\epsilon \norm{\ha-\ha_k}_F^2}{\sqrt{nk}}.
\end{align*}
Which in turn implies
\begin{align*}
\sigma_1^2(S^T\ha (I-U_TU_T^T)) \le (1+\epsilon)\sigma_1^2(\ha (I-U_TU_T^T)) +  \frac{\epsilon \norm{\ha-\ha_k}_F^2}{\sqrt{nk}}.
\end{align*}
and so, applying the AMGM inequality,
\begin{align*}
\norm{S_2^T C_T}_2^2 &= \norm{S_2^T (I-U_TU_T^T) AS_1}_2^2\\
 &\le \sigma_1^2(S_2^T \ha (I-U_T U_T^T) ) \cdot \sigma_1^2((I-U_T U_T^T)\ha  S_1)\\
&\le 2(1+\epsilon)^2 \sigma_1^4(\ha (I-U_TU_T)^T) + \frac{2\epsilon^2 \norm{\ha-\ha_k}_F^4}{nk}\\
&\le 8 \norm{A(I-U_TU_T^T)}_2^2 + \frac{2\epsilon^2 \norm{A-A_k}_F^2}{k}\tag{by $\ell_1$/$\ell_2$ bound.}\\
&\le \frac{10\epsilon^2 \norm{A-A_k}_F^2}{k}.\tag{$\norm{A(I-U_TU_T^T)}_2^2  \le \frac{\epsilon^2 \norm{A-A_k}_F^2}{k}$ by definition.}
\end{align*}
Similarly we have:
\begin{align*}
\norm{C_T}_2^2 &= \sigma_1^2((I-U_TU_T^T)AS_1)\\
&\le \sigma_1^2(S_1^T \ha (I-U_TU_T^T)) \cdot \sigma_1^2(\ha (I-U_TU_T^T))\\
&\le (1+\epsilon) \sigma_1^4(\ha (I-U_TU_T^T)) + \frac{\epsilon \norm{\ha-\ha_k}_F^2}{\sqrt{nk}} \cdot \sigma_1^2(\ha (I-U_TU_T^T))\\
&\le \frac{3\epsilon^2 \norm{A-A_k}_F^2}{k}.
\end{align*}
So overall, plugging back into \eqref{tailLoose} gives:
\begin{align}
\norm{(I-P)(C_T^T S_2 S_2^T C_T - C_T^T C_T)(I-P) }_2 \le \frac{13\epsilon^2}k  \norm{A-A_k}_F^2.\label{spectralTail}
\end{align}

\medskip
\noindent
\textbf{Cross Term:}
\medskip

By submultiplicativity of the spectral norm:
\begin{align*}
\norm{(I-P) C_{T}^T S_2 S_2^T C_H (I-P)}_2 \le \norm{C_{T}^T S_2}_2 \cdot \norm{S_2^T C_H (I-P)}_2.
\end{align*}
As shown above for our tail bound, $\norm{C_{T}^T S_2}_2 \le \sqrt{\frac{10\epsilon^2}{k}} \norm{A-A_k}_F$. Further, as shown in \eqref{Csubspace}, $S_2$ gives a subspace embedding for $C_H$ so we have:
\begin{align}
\norm{(I-P) C_{T}^T S_2 S_2^T C_H (I-P)}_2 &\le \sqrt{\frac{10\epsilon^2}{k}} \norm{A-A_k}_F \cdot (1+\epsilon) \norm{C_H(I-P)}_2\nonumber\\
&\le \frac{10 \epsilon}{k} \norm{A-A_k}_F^2 + 4\epsilon \norm{C(I-P)}^2_2\label{spectralCross}.
\end{align}
where the last bound follows form the AMGM inequality.

Plugging our head \eqref{spectralHead}, tail \eqref{spectralTail} and cross term \eqref{spectralCross} bounds in \eqref{spectralBreakdown}, and adjusting constants on $\epsilon$ by making $c$ and $c'$ sufficiently large gives the lemma.

\end{proof}


\begin{replemma}{leveragePCP}[Frobenius Norm Row PCP] For any PSD $A \in \R^{n \times n}$ and $\epsilon \le 1$ let $k' = \lceil ck/\epsilon \rceil$ and let $\tilde \tau_i^{k'}(\ha) \ge \tau_i^{k'}(\ha)$ be an overestimate for the $i^{th}$ rank-$k'$ ridge leverage score of $\ha$. Let $\tilde \ell_i = \sqrt{\frac{16n\epsilon}{k}} \cdot \tilde \tau_i^{k'}(\ha) $, $p_i = \frac{\tilde \ell_i}{\sum_i \tilde \ell_i}$, and $t = \frac{c'\log n}{\epsilon^2} \sum_i \tilde \ell_i$.
Construct weighted sampling matrices $S_1,S_2 \in \R^{n \times t}$ each whose $j^{th}$ column is set to $\frac{1}{\sqrt{tp_i}} e_i$ with probability $p_i$.

For sufficiently large constants $c,c'$, with probability $\frac{99}{100}$, letting $\tilde A = S_2^T A S_1$, for any rank-$k$ orthogonal projection $P \in \mathbb{R}^{t \times t}$:
\begin{align*}
(1-\epsilon) \norm{AS_1(I-P)}_F^2 \le \norm{\tilde A(I-P)}_F^2 + \Delta \le (1+\epsilon) \norm{AS_1(I-P)}_F^2
\end{align*}
 for some fixed $\Delta$ (independent of $P$) with $|\Delta| \le 600 \norm{A-A_k}_F^2$.
\end{replemma}
Note that if $\tilde \tau_i^{k'}(\ha)$ is a constant factor approximation to $\tau_i^{k'}(\ha)$, $t = O \left (\frac{\sqrt{nk}\log n}{\epsilon^{2.5}} \right)$.
\begin{proof}
The proof is similar to that of Lemma \ref{leveragePCPspectral}.
Denote $C \eqdef AS_1$ and write the eigendecomposition $A=U\Lambda U^T$ with $\lambda_i = \Lambda_{i,i}$. Let $m$ be the largest index with $\lambda_m^2 \ge \frac{\epsilon}{k} \norm{A-A_{k'}}_F^2.$ Let $U_H \eqdef U_m$ contain the top $m$ `head' eigenvectors of $A$ and let $U_T$ contain the remaining `tail' eigenvectors. Let $C_H = U_HU_H^T C$ and $C_T = U_T U_T^T C$ and note that $C_H + C_T = C$. 

By the Pythagorean theorem we have $\norm{C(I-P)}_F^2 = \norm{C_H (I-P)}_F^2 + \norm{C_{T}(I-P)}_F^2$.
Expanding using the identity $\norm{M}_F^2 = \tr(M^T M)$,
\begin{align}
\norm{\tilde A(I-P)}_F^2 &= \norm{S_2^T C_{H}(I-P)}_F^2 + \norm{S_2^T C_{T}(I-P)}_F^2 + 2\tr \left ( (I-P) C_H^T S_2 S_2^T C_T (I-P)\right)
\label{headTailSplit}
\end{align}
We bound each of the terms in the above sum separately. Specifically we show:
\begin{itemize}
\item Head Term: $\norm{S_2^TC_H(I-P)}_F^2 \in (1 \pm \epsilon) \norm{C_H(I-P)}_F^2$
\item Tail Term:  $\norm{S_2^TC_{T}(I-P)}_F^2 + \Delta \in \norm{C_{T}(I-P)}_F^2 \pm \epsilon \norm{A-A_k}_F^2$ where $|\Delta| \le 600\norm{A-A_k}_F^2$.
\item Cross Term: $|\tr\left ((I-P)C_{H}^T S_2 S_2^T C_T (I-P)\right )| \le \epsilon \norm {C(I-P)}_F^2.$
\end{itemize}
It is not hard to see that combining these bounds gives the Lemma.
$\norm{A-A_k}_F^2 \le (1+3\epsilon) \norm{C(I-P)}_F^2$ since $C$ is an $(\epsilon,k)$-column PCP of $A$ by Lemmas \ref{thm:pcp}, \ref{lem:scoreBound}, and the fact that rank-$k'$ ridge scores upper bound the rank-$k$ scores (as long as $c > 1$). Additionally,
$\norm{A-A_k}_F^2 \le \norm{A(I-P)}_F^2$ for any rank-$k$ $P$. So plugging into \eqref{headTailSplit} we have:
\begin{align*}
\norm{\tilde A(I-P)} + \Delta &\in (1\pm \epsilon) \norm{C(I-P)} \pm (2+3\epsilon)\epsilon \norm{C(I-P)}_F^2\\
&\in (1\pm O(\epsilon) ) \norm{C(I-P)}.
\end{align*}
This gives the lemma by adjusting constants on $\epsilon$ by making $c$ and $c'$ large enough.  For the remainder of the proof we thus fix $c = 1$ and so $k' = \lceil k/\epsilon \rceil$.

\medskip
\noindent
\textbf{Head Term:}
\medskip

By Lemma \ref{lem:aux1} applied to $k' = \lceil k/\epsilon \rceil$, since $\lambda_m^2 \ge \frac{1}{k'} \norm{A-A_{k'}}_F^2$:
$$\tilde \ell_i \ge \sqrt{\frac{16n\epsilon}{k}} \cdot  \tilde \tau_i^{k'}(\ha) \ge \norm{(U_H)_i}_2^2.$$

Further, $t = \frac{c'\log n}{\epsilon^2 } \cdot \sum_i \tilde \ell_i $ so if we set $c'$ large enough, by a standard matrix Chernoff bound (see Lemma \ref{lem:chernoff}) since $U_H$ is an orthogonal span for the columns of $C_H$ and hence and its row norms are the leverage scores of $C_H$, with high probability:
\begin{align*}
(1-\epsilon) C_H^T C_H \preceq C_H^TS_2S_2^TC_H \prec (1+\epsilon)C_H^TC_H.
\end{align*}
This gives the bound $\norm{S_2^TC_H(I-P)}_F^2 \in (1 \pm \epsilon) \norm{C_H(I-P)}_F^2$.

\medskip
\noindent
\textbf{Tail Term:}
\medskip
We want to show:
\begin{align}\label{tailBound}
\norm{S_2^TC_{T}(I-P)}_F^2 + \Delta \in \norm{C_{T}(I-P)}_F^2 \pm \epsilon \norm{A-A_k}_F^2
\end{align}
where $|\Delta| \le 600\norm{A-A_k}_F^2$.

We again split using Pythagorean theorem, $\norm{S_2^TC_{T}(I-P)}_F^2  = \norm{S_2^TC_T}_F^2 - \norm{S_2^T C_T P}_F^2$. We set $\Delta = \norm{C_T}_F^2 - \norm{S_2^T C_T}_F^2$. We have $\E \norm{S_2 C_T}_F^2 = \norm{C_T}_F^2 \le (1+\epsilon) \norm{A-A_k}_F^2$, where the last bound follows since $C$ is an $(\epsilon,k)$-column PCP for $A$. Thus by a Markov bound, with probability $299/300$, $|\Delta| \le (1+\epsilon)300\norm{A-A_k}_F^2 \le 600 \norm{A-A_k}_F^2$. 

Additionally,
we have $\norm{C_T}_2^2 \le \frac{10\epsilon}{k} \norm{A-A_k}_F^2$ and $\norm{S_2^T C_T}_2^2 \le \frac{10\epsilon}{k} \norm{A-A_k}_F^2$ with high probability by an identical argument to that used for Lemma \ref{leveragePCPspectral}. This gives:
\begin{align*}
\left |\norm{S_2^TC_TP}_F^2 - \norm{C_TP}_F^2 \right | \le k (\norm{S_2^T C_T}_2^2 + \norm{C_T}_2^2) \le 20\epsilon \norm{A-A_k}_F^2
\end{align*}
which gives the main bound \eqref{tailBound} after adjusting constants on $\epsilon$.

\medskip
\noindent
\textbf{Cross Term:}
\medskip
We want to show:
\begin{align}\label{cross1Bound}
|\tr\left ((I-P)C_{T}^T S_2 S_2^TC_H (I-P)\right )| \le \epsilon \norm {C(I-P)}_F^2.
\end{align}
We can write:
\begin{align*}
\left |\tr\left ((I-P)C_{T}^T S_2 S_2^T C_H (I-P)\right ) \right | &= \left |\tr\left (C_{T}^T S_2 S_2^T C_H (I-P)\right )\right|\tag{Cyclic property of trace and $(I-P) = (I-P)^2$}\\
&= \left |\tr \left (C_{T}^T S_2 S_2^TC_{H}(C^TC)^+(C^TC)(I-P) \right )\right|
\end{align*}
where in the last step, inserting $(C^TC)^+(C^TC)$, which is the projection onto the row span of $C$ has no effect as the rows of $C_H = U_H U_H^T C$ already lie in this span. $\langle M, N \rangle = \tr( M (C^T C)^+ N^T)$ is a semi-inner product since $C^TC$ is positive semidefinite, so by Cauchy-Schwarz:
\begin{align*}
\left |\tr\left ((I-P)C_{T}^T S_2 S_2^TC_{H} (I-P)\right ) \right | \le \norm{C_{T}^T S_2 S_2^T C_{H} (C^TC)^{+/2} }_F \cdot \norm{C(I-P)}_F.
\end{align*}
Using the singular value decomposition $C = XS Y^T$:
\begin{align}
\left |\tr\left (C_{T}^T S_2 S_2^TC_{H} (I-P)\right ) \right | &\le \norm{C_{T}^T S_2 S_2^T U_HU_H ^T X}_F \cdot \norm{C(I-P)}_F\nonumber\\
& \le \norm{C_{T}^TS_2S_2^T U_H}_F \cdot \norm{C(I-P)}_F\label{firstTrace}
\end{align}
As argued, by Lemma  \ref{lem:aux1} we have $\tilde \ell_i \ge \norm{(U_H)_i}_2^2$, and so by a standard approximate matrix multiplication result \cite{drineas2006fast}, with probability $299/300$ if we set the constant $c'$ in our sample size $> c''$ for some fixed $c''$:
\begin{align*}
\norm{C_{T}^T S_2 S_2^T U_H}_F &\le \frac{\norm{U_H}_F \norm{C_{T}}_F}{\sqrt{\frac{t \cdot c'' \norm{U_H}_F^2}{\sum_i \tilde \ell_i} \cdot }}\\
&\le \epsilon \norm{C_{T}}_F = O(\epsilon \norm{A-A_k}_F)
\end{align*}
where the last step follows from the fact that $\norm{C_T}_F^2 = O(\norm{A-A_k}_F^2)$ since $C$ is an $(\epsilon,k)$-column PCP for $A$. The final bound then follows from combining with \eqref{firstTrace} with the fact that $\norm{A-A_k}_F^2 \le (1+\epsilon) \norm{C(I-P)}_F^2$ and adjusting constants on $\epsilon$ by increasing the constant $c'$ in the sample size $t$ and $c$ in the rank parameter $k' = \lceil ck/\epsilon \rceil$.

The full lemma follows simply from noting that we can union bound over our failure probability for each term so all hold simultaneously with probability $99/100$.
\end{proof}

\begin{lemma}[Leverage Score Sampling Matrix Chernoff]\label{lem:chernoff} For any $A \in \R^{n \times d}$
for $i \in \{1,\ldots,d\}$, let $\ell_i(A) \eqdef a_i^T (AA^T)^+ a_i$ be the $i^{th}$ column leverage score of $A$ and let $\tilde \ell_i \ge \ell_i(A)$ be an overestimate for this score. 
Let $p_i = \frac{\tilde \ell_i}{\sum_i \tilde \ell_i}$ and $t = \frac{c\log(d/\delta)}{\epsilon^2} \sum_i \tilde \ell_i$ for sufficiently large $c$. Construct $C$ by sampling $t$ columns of $A$, each set to $\frac{1}{\sqrt{tp_i}}a_i$ with probability $p_i$. With probability $1-\delta$:
\begin{align}\label{chernoff_main_bound}
(1-\epsilon) CC^T \preceq AA^T \preceq (1+\epsilon) CC^T.
\end{align}
\end{lemma}
\begin{proof}
Write the singular value decomposition $A = U\Sigma V^T$. Note that:
$$\ell_i = e_i^T V \Sigma U^T \left (U\Sigma^2 U^T \right)^+ U \Sigma V^T e_i = \norm{v_i}_2^2.$$

Let $Y =  \Sigma^{-1}U^T\left (CC^T - AA^T \right)U \Sigma^{-1}$. We can write:
$$Y = \sum_{j=1}^t \left [ { \Sigma}^{-1} U^T\left (c_j c_j^T-\frac{1}{t}AA^T\right ) U  \Sigma^{-1}\right ]\eqdef \sum_{j=1}^t \left [ X_j \right ].$$
For each $j \in 1,\ldots,t$, $X_j$ is given by:
\begin{align*}
X_j = 
\frac{1}{t} \cdot \Sigma^{-1}U^T\left (\frac{1}{p_i} a_i a_i^T - AA^T\right ) U \Sigma^{-1} \text{ with probability } p_i.
\end{align*}

$\E Y = 0$ since $\E X_j = \sum_{i=1}^d p_i \left [ \frac{1}{p_i} a_i a_i^T - AA^T\right ] = 0$. Furthermore, $CC^T = U \Sigma Y \Sigma U + AA^T$. Showing $\norm{Y}_2 \le \epsilon$ gives $-\epsilon I \preceq Y \preceq  \epsilon I$, which gives the lemma. We prove this bound using 
a matrix Bernstein inequality from \cite{tropp2015introduction}. This inequality requires upper bounds on the spectral norm of each $X_j$ and on variance of $Y$. We first note that 
for any $i$, $\frac{1}{\ell_i(A)} a_i a_i^T \preceq AA^T.$ This follows from writing any $x$ in the column span of $A$ as $(AA^T)^{+/2}y$ and then noting:
\begin{align*}
x^T \left (a_i a_i^T \right ) x = y^T (AA^T)^{+/2} a_i a_i^T (AA^T)^{+/2} y \le \ell_i(A) \norm{y}_2^2
\end{align*}
since $(AA^T)^{+/2} a_i a_i^T (AA^T)^{+/2}$ is rank-$1$ and so has maximum eigenvalue $\tr\left ((AA^T)^{+/2} a_i a_i^T (AA^T)^{+/2} \right ) = \ell_i(A)$ 
by the cyclic property of trace. Further $x^T AA^T x = y ^T (AA^T)^{+/2} AA^T (AA^T)^{+/2} y = \norm{y}_2^2$, giving us the bound. We then have:
\begin{align*}
\frac{1}{\ell_i(A)} \cdot \Sigma^{-1} U^Ta_ia_i^T U\Sigma^{-1} \preceq \Sigma^{-1} U^TAA^T U\Sigma^{-1} = I.
\end{align*}
And so $\frac{1}{tp_i}\Sigma^{-1} U^Ta_ia_i^T U\Sigma^{-1} \preceq \frac{\epsilon^2}{c\log(d/\delta ) \tilde \ell_i} \Sigma^{-1} U^Ta_ia_i^T U\Sigma^{-1} \preceq \frac{\epsilon^2}{c\log(d/\delta)} I$. Additionally, $$\frac{1}{tp_i}\Sigma^{-1} U^TAA^T U\Sigma^{-1} \preceq \frac{\epsilon^2}{c\log(d/\delta)} I$$ as long as $\tilde \ell_i \le 1$ which we may as well enforce since $\ell_i(A) \le 1$.
Overall this gives $\norm{X_j}_2 \le \frac{\epsilon^2}{c\log(d/\delta)}$.
Next we bound the variance of $Y$.
\begin{align}
\E (Y^2) &= t \cdot \E (X_j^2 ) = \frac{1}{t} \sum p_i \cdot  \left (\frac{1}{p_i^2} \Sigma^{-1} U^Ta_i a_i^T U \Sigma^{-2} U^Ta_ia_i^T U \Sigma^{-1} \right. \nonumber\\
&\left. - 2\frac{1}{p_i} \Sigma^{-1} U^T a_i a_i^T U \Sigma^{-2} U^TAA^T U {\Sigma}^{-1} + \Sigma^{-1} U^TAA^T U \Sigma^{-2} U^TAA^T U \Sigma^{-1} \right) \nonumber\\
&\preceq \frac{1}{t}\sum \left [ \frac{\sum \tilde\ell_i}{\tilde \ell_i} \cdot  \ell_i(A) \cdot \Sigma^{-1} Ua_i a_i^T U \Sigma^{-1}\right ] - \frac{1}{t}\Sigma ^{-1} U^T AA^T U\Sigma^{-1}\nonumber\\
&\preceq \frac{\epsilon^2}{c\log(d/\delta)} \Sigma ^{-1} U^T AA^T U\Sigma^{-1} \nonumber \preceq  \frac{\epsilon^2}{c\log(d/\delta)} I.
\end{align}
By Theorem 7.3.1 of \cite{tropp2015introduction}, for $\epsilon < 1$,
\begin{align*}
\Pr \left [\norm{ Y}_2 \ge \epsilon \right ] &\le 4 d \cdot e^{\frac{-\epsilon^2/2}{\left (\frac{\epsilon^2}{c\log(d/\delta)}(1+\epsilon/3)\right )}}.
\end{align*}
Which gives $\Pr \left [\norm{ Y}_2 \ge \epsilon \right ]  \le 4d e^{- \frac{c\log(d/\delta)}{4}}\le \delta$
if we choose $c$ large enough. 
\end{proof}

The above Lemma also gives an easy Corollary for the approximation obtained when sampling with ridge leverage scores:
\begin{corollary}[Ridge Leverage Scores Sampling Matrix Chernoff]\label{cor:chernoff}
For any $A \in \R^{n \times d}$
for $i \in \{1,\ldots,d\}$, let $\tau^\lambda_i(A) \eqdef a_i^T (AA^T + \lambda I)^+ a_i$ be the $i^{th}$ $\lambda$-ridge leverage score of $A$ and let $\tilde \tau_i^\lambda \ge \tau_i^\lambda(A)$ be an overestimate for this score. 
Let $p_i = \frac{\tilde \tau_i^\lambda}{\sum_i \tilde \tau_i^\lambda}$ and $t = \frac{c\log(d/\delta)}{\epsilon^2} \sum_i \tilde \tau_i^\lambda$ for sufficiently large $c$. Construct $C$ by sampling $t$ columns of $A$, each set to $\frac{1}{\sqrt{tp_i}}a_i$ with probability $p_i$. With probability $1-\delta$:
\begin{align}
(1-\epsilon) CC^T - \epsilon \lambda I \preceq AA^T \preceq (1+\epsilon) CC^T + \epsilon \lambda I.
\end{align}
\end{corollary}
\begin{proof}
We can instantiate Lemma \ref{lem:chernoff} with $[ A, \sqrt{\lambda I}]$ setting $\tilde \ell_i = \tilde \tau_i^\lambda$. We simply fix the columns of the identity to appear in our sample. This only decreases variance, all calculations go through, and with probability $\delta$:
\begin{align*}
(1-\epsilon) [C,\sqrt{\lambda} I][C,\sqrt{\lambda} I]^T \preceq [A,\sqrt{\lambda} I][A,\sqrt{\lambda} I]^T \preceq (1+\epsilon) [C,\sqrt{\lambda} I][C,\sqrt{\lambda} I]^T
\end{align*}
which gives the desired bound if we subtract $\lambda I$ from all sides.
\end{proof}

\section{Outputting a PSD Low-Rank Approximation}\label{psd_appendix}

In this section, we prove Theorem \ref{thm:psdOutput}, showing to to efficiently output a low-rank approximation to $A$ under the restriction that the approximation is also PSD. We start with the following lemma, which shows that as long as we have a good low rank subspace for approximating $A$ in the spectral norm (computable using Theorem \ref{thm:spectral}), we can quickly find a near optimal PSD low-rank approximation:

\begin{lemma}\label{psdApproxProj}
Given PSD $A \in \R^{n \times n}$ and orthonormal basis $Z \in R^{n \times m}$ with $\norm{A-AZZ^T}_2^2 \le \frac{\epsilon}{k}\norm{A-A_k}_F^2$, there is an algorithm which accesses $O \left(\frac{nm\log m}{\epsilon^2} \right)$ entries of $A$, runs in $\tilde O \left (\frac{nm^{\omega-1} }{\epsilon^{2(\omega-1)}} \right )$ time, and with probability $99/100$ outputs $M \in \mathbb{R}^{n \times k}$ satisfying $\norm{A - MM^T}_F^2 \le (1+\epsilon) \norm{A-A_k}_F^2$.
\end{lemma}
\begin{proof}
By Lemma 10 of \cite{CWPSD} for any basis $Z \in \R^{n \times m}$ with $\norm{A-AZZ^T}_2^2 \le \frac{\epsilon}{k} \norm{A-A_k}_F^2$,
\begin{align*}
\min_{X: rank(X) = k, X \succeq 0} \norm{A-Z X Z^T}_F^2 \le (1+O(\epsilon))\norm{A-A_k}_F^2.
\end{align*}

As in \nameref{mainAlgo} we can find a near optimal $X$ by further sampling $Z$ using its leverage scores. If we sample $t_1 = O \left (\frac{m \log m}{\epsilon^2}\right)$ rows of $Z$ by their leverage scores (their norms since $Z$ is orthonormal) to form $S_1 \in R^{n \times t_1}$, by Theorem 39 of \cite{clarkson2013low}, we will have an \emph{affine embedding} of $Z$. Specifically, letting $B^* = \argmin_B \norm{A - ZB}_F^2$ and $E^* = A-ZB^*$, for any $B$ we have:
\begin{align*}
\norm{S_1^T A - S_1 Z B}_F^2 + \left (\norm{E^*}_F^2 - \norm{S_1^T E^*}_F^2 \right ) \in \left [ (1-\epsilon) \norm {A-ZB}_F^2, (1+\epsilon) \norm{A-ZB}_F^2 \right ].
\end{align*}
Note that this is similar to the embedding property used in the proof of Theorem \ref{thm:main} to show that $W$ computed in Step 6 of \nameref{mainAlgo} gave a near optimal low-rank approximation to $AS_1$.

By a Markov bound, since $\E \norm{S_1^T E^*}_F^2 = \norm{E^*}_F^2$, with probability $99/100$, $\left | \norm{E^*}_F^2 - \norm{S_1^T E^*}_F^2 \right | \le 100\norm{E^*}_F^2 = O(1) \norm{A-A_k}_F^2$. This guarantees that a $(1+\epsilon)$ approximation to the sketched problem gives a $(1+O(\epsilon))$ approximation to the original. That is, for any PSD $\tilde X$ with $\rank(\tilde X) = k$, and 
$$\norm{S_1^TA - S_1 Z \tilde X Z^T}_F^2 \le (1+\epsilon)\min_{X: \rank(X) = k, X \succeq 0}  \norm{S_1^TA-S_1ZXZ^T}_F^2$$
we have $\norm{A-Z\tilde X Z^T}_F^2 \le (1+O(\epsilon)) \norm{A-A_k}_F^2.$ 
Following \cite{CWPSD}, we write $S_1^T Z$ in its SVD $S_1^TZ = U_z\Sigma_z V_z^T$. Since $S_1ZX Z^T$ falls in the column span of $S_1Z$ and the row span of $Z^T$, we write $\delta \eqdef \norm{(I-U_zU_z^T)S_1^TA}_F^2 + \norm{U_zU_z^TS_1^TA(I-ZZ^T)}_F^2$ and by Pythagorean theorem have:
\begin{align*}
\norm{S_1^TA-S_1ZXZ^T}_F^2 &= \norm{U_zU_z^TS_1 AZZ^T - U_z \Sigma _z V_z^T X Z^T}_F^2 + \delta\\
&= \norm{U_z^T S_1^T AZ - \Sigma_z V_z^T X}_F^2 + \delta \\
&= \norm{\Sigma_z V_z^T \left (V_z \Sigma_z^{-1} U_z^T S_1^T AZ - X \right)}_F^2 + \delta\end{align*}
Since $S_1$ is sampled via $Z$'s leverage scores, $(1-\epsilon) I \preceq S_1^T Z \preceq (1+\epsilon ) I$ and so:
$$\norm{S_1^TAZ-S_1ZXZ^T}_F^2 = (1\pm \epsilon) \norm{V_z^T \Sigma_z^{-1} U_z^T S_1^T AZ - X}_F^2 + \delta.$$
Finally, following \cite{CWPSD}, letting $B = V_z^T \Sigma_z^{-1} U_z^T S_1^T AZ$, we have 
$$\tilde X = \argmin_{X | X\succeq 0, \rank(X) = k} \norm{B - X}_F^2 = \left (B/2 + B^T/2 \right)_{k,+}$$
where $N_{k,+}$ has all but the top $k$ \emph{positive} eigenvalues of $N$ set to $0$. We output $M = Z \tilde X^{1/2}$. Overall, the above algorithm requires accessing $O \left (\frac{nm\log m}{\epsilon^2} \right)$ entries of $A$ (the entries of $AS_1$) and has runtime $\tilde O \left(\frac{n m^{\omega-1}}{\epsilon^{2(\omega-1)}}  \right )$ giving the lemma.
\end{proof}

We  can obtain $Z$ with rank $m = \Theta(k/\epsilon)$ by applying Theorem \ref{thm:spectral} with rank $k' = \Theta(k/\epsilon)$ and error parameter $\epsilon' = \Theta(1)$.
Combined with Lemma \ref{psdApproxProj} this yields:
\begin{reptheorem}{thm:psdOutput}[Sublinear Time Low-Rank Approximation -- PSD Output]
There is an algorithm that given any PSD $A \in \R^{n \times n}$, accesses $\tilde O \left (\frac{nk^2}{\epsilon^2} + \frac{nk}{\epsilon^3} \right)$ entries of $A$, runs in $\tilde O \left (\frac{nk^\omega}{\epsilon^\omega} + \frac{n k^{\omega-1}}{\epsilon^{3(\omega-1)}} \right)$ time and with probability at least $9/10$ outputs $M \in \R^{n \times k}$ with:
$$\norm{A-MM^T}_F^2 \le (1+\epsilon)\norm{A-A_k}_F^2.$$
\end{reptheorem}

\end{document}